\documentclass[11pt]{article}
\usepackage[margin=1in]{geometry}
\usepackage{amsmath, amssymb, amsthm}
\usepackage{xcolor}
\usepackage{parskip}
\usepackage{algorithm}
\usepackage{algpseudocode}
\usepackage{hyperref, comment, graphicx}

\newtheorem{theorem}{Theorem}
\newtheorem{lemma}{Lemma}
\newtheorem{claim}{Claim}
\newtheorem{corollary}{Corollary}
\newtheorem{proposition}{Proposition}
\theoremstyle{definition}
\newtheorem{definition}{Definition}
\theoremstyle{remark}
\newtheorem{remark}{Remark}

\newcommand{\cB}{\mathcal{B}}
\newcommand{\cL}{\mathcal{L}}
\newcommand{\cT}{\mathcal{T}}
\newcommand{\cQ}{\mathcal{Q}}
\newcommand{\cW}{\mathcal{W}}
\newcommand{\cZ}{\mathcal{Z}}
\newcommand{\row}{\mathrm{row}}
\newcommand{\subspan}{\mathrm{span}}
\newcommand{\numberofblocks}{t}
\newcommand{\R}{\mathbb{R}}
\newcommand{\OPT}{\mathrm{OPT}}

\newcommand{\orient}{\mathrm{orient}}
\newcommand{\ESD}{\mathrm{Edge-SD}}
\newcommand{\SD}{\mathrm{SD}}

\newcommand{\Fset}{\mathcal{F}}

\newcommand{\NZ}{NZ}
\newcommand{\nnz}{p}

\def\final{0}  
\def\iflong{\iffalse}
\ifnum\final=0  
\newcommand{\cnote}[1]{{\color{red}[{Chandra: \bf #1}]\marginpar{*}}}
\newcommand{\knote}[1]{{\color{blue}[{Karthik: \bf #1}]\marginpar{\color{blue}*}}}
\newcommand{\todo}[1]{{\color{red}[{\tiny TODO: \bf #1}]\marginpar{\color{red}*}}}
\else 
\newcommand{\cnote}[1]{}
\newcommand{\knote}[1]{}
\newcommand{\todo}[1]{}
\fi  

\title{An iterative rounding $2$-approximation for Feedback Vertex Set
  via AI-assisted proof of an extreme point property\thanks{Grainger College of Engineering, University of Illinois, Urbana-Champaign, Email: {\tt\{karthe, chekuri, smkulka2\}@illinois.edu}. Supported in part by NSF grant CCF-2402667. Work done while Shubhang was a student at UIUC.}}
\author{Karthekeyan Chandrasekaran \and Chandra Chekuri \and Shubhang Kulkarni}
\date{}

\begin{document}

\maketitle

\begin{abstract}
We consider the Feedback Vertex Set problem (FVS): the input is an undirected graph $G=(V,E)$ and the goal is to find a minimum-cardinality (or a min-cost in the weighted case) subset $S \subseteq V$ of vertices such that $G-S$ has no cycles.  A $2$-approximation via the local-ratio method was developed in the mid 90's by Bafna, Berman and Fujito \cite{Bafna-Berman-Fujito95} and by Becker and Geiger \cite{BG96}, and this approximation ratio is tight under UGC. The local-ratio algorithms were later interpreted as primal-dual algorithms via an LP relaxation by Chudak, Goemans, Hochbaum, and Williamson \cite{Chudaketal}. All known $2$-approximation algorithms for FVS have been via local-ratio and primal-dual methods, and in a quest to obtain a new LP rounding algorithm, it was conjectured \cite{Fiorini21,CCFKW25} that the Strong-Density polyhedron developed in \cite{Chudaketal} has an extreme point property: every basic feasible solution to the LP has a variable with value at least $1/2$. We prove this conjecture. We also consider a related Edge-Strong-Density polyhedron and show the same extreme point property. The advantage of this polyhedron is that it admits a polynomial-time separation oracle and also a compact extended formulation. These results lead to polynomial-time iterative rounding $2$-approximation algorithms. The proof of the extreme point property is of independent technical interest and key ideas in the proof were suggested by AI tools. 


\end{abstract}

\newpage
\setcounter{page}{1}
\section{Introduction}
Given a graph $G=(V,E)$ a \emph{feedback vertex set} for $G$ is a subset  of vertices $S \subseteq V$ whose removal makes the graph acyclic. In other words, $S$ is a hitting set for the cycles of the graph. The Feedback Vertex Set problem (FVS) is defined as follows: given a graph $G = (V, E)$ with non-negative vertex costs $c: V \rightarrow \mathbb{R}_{\ge 0}$, find a least-cost feedback vertex set. FVS is a classical combinatorial optimization problem and was shown to be NP-Hard in Karp's well-known paper on NP-Completeness \cite{Karp72}. FVS is also interesting in graph theory. The well-known Erd\"os-P\'osa theorem \cite{ErdosP62} shows that there is a feedback vertex set whose cardinality is $O(\log k^*) \cdot k^*$ where $k^*$ is the maximum number of vertex-disjoint cycles in $G$. Moreover, this bound is tight in the worst case, for instance, in a constant degree expander graph on $n$ nodes in which $k^* = O(n/\log n)$ and the minimum feedback vertex set size is $\Omega(n)$. 

In this work, we are interested in approximation algorithms for FVS via LP rounding.  The Erd\"os-P\'osa theorem implicitly gives an $O(\log n)$-approximation for min-cardinality FVS, and the lower bound also shows that a natural hitting set based LP relaxation has an $\Omega(\log n)$-factor integrality gap (this was first explicitly pointed out in \cite{BarYehudaGNR98}). Independent works of Bafna, Berman, and Fujito \cite{Bafna-Berman-Fujito95} and Becker and Geiger \cite{BG96}, obtained $2$-approximation algorithms for FVS in the mid 90's. These algorithms were combinatorial and were explicitly or implicitly based on the local-ratio method. 
An $\alpha$-approximation for FVS implies an $\alpha$-approximation for the Vertex Cover problem, and hence, under the known hardness for Vertex Cover under the Unique Games Conjecture \cite{KhotR08}, we do not expect a $(2-\epsilon)$-approximation for FVS.
Chudak, Goemans, Hochbaum, and Williamson \cite{Chudaketal} described exponential-sized integer linear programming (ILP) formulations for FVS, and interpreted the algorithms in \cite{Bafna-Berman-Fujito95,BG96} as primal-dual algorithms with respect to these LP relaxations. This also established an upper bound of $2$ on the integrality gap of these relaxations. Despite these developments, important caveats remained. The LP relaxations in \cite{Chudaketal} were not known to be solvable in polynomial time, and in fact, no explicit polynomial-time solvable LP relaxations with a constant factor integrality gap were known for a long time until \cite{chekuri-madan16}. Fiorini \cite{Fiorini21}, motivated by a desire to obtain new algorithms for FVS and generalizations, conjectured that the LP relaxation in \cite{Chudaketal} has an extreme point property that may lead to an iterative rounding $2$-approximation. This conjecture inspired Chandrasekaran, Chekuri, Fiorini, Kulkarni, and Weltge \cite{CCFKW25} to do a polyhedral investigation of FVS. In particular, \cite{CCFKW25} showed that several different LP formulations for FVS, all of which can be solved in polynomial time, have an integrality gap of $2$. Nevertheless, the original conjecture of Fiorini was not proven, and the integrality gap and the $2$-approximation results in \cite{CCFKW25} still relied on a primal-dual analysis.

In this paper, we prove the conjecture in \cite{Fiorini21,CCFKW25} regarding the extreme point property, and as a consequence, we derive new 2-approximation algorithms for FVS based on iterative rounding.

We consider two formulations for FVS. They have indicator variables $(x_u)_{u \in V}$ for whether a vertex $u$ is in the feedback vertex set and an exponential number of constraints. We set up some basic notation. For a vertex set $S \subseteq V$ of a graph $G=(V, E)$, let $G[S]$ denote the subgraph induced by $S$, let $E[S]$ denote its edge set, and let $d_S(u)$ denote the degree of vertex $u$ in $G[S]$. 

\begin{definition}[Strong Density Polyhedron]
Let $G = (V, E)$ be an undirected graph.
The \emph{strong density polyhedron} $P_{SD}(G)$ is defined as
\[
P_{SD}(G) := \left\{ x\in [0,1]^V \;:\; \sum_{u \in S} (d_S(u) - 1)\, x_u \;\ge\; |E[S]| - |S| + 1 \quad \forall\, S \subseteq V \text{ with } E[S] \ne \emptyset \right\}.
\]

For an edge set $F\subseteq E$ of a graph $G=(V, E)$, let $V(F)$ denote the set of vertices incident to at least one edge in $F$, and we write $G[F]$ to denote the subgraph $(V(F), F)$. We overload notation and let $d_F(u)$ denote the degree of vertex $u$ in the subgraph $(V(F), F)$. For the most part, the notation overload will not cause any confusion. We will alert the reader when we overload.

\end{definition}

\begin{definition}[Edge Strong Density Polyhedron]
Let $G = (V, E)$ be an undirected graph.
The \emph{edge strong density polyhedron} $P_{\ESD}(G)$ is defined as:
\[
P_{\ESD}(G):=\left\{
x\in [0,1]^V: 
\sum_{u \in V(F)} (d_F(u) - 1) x_u \ge |F| - |V(F)| + 1 \quad \forall \emptyset\neq F \subseteq E\right\}. 
\]
  
\end{definition}

The definition of $P_{\ESD}(G)$ imposes constraints for all non-empty edge-subsets (including acyclic edge-subsets). 
The use of the term ``density'' in the names of the formulation is due to connection to the densest subgraph problem that arose in recent work \cite{Fiorini21,CCFKW25}, and will be explained later. We now explain the meaning of the constraints and the relationship between the two formulations. The formulation $P_{SD}(G)$ is from \cite{Chudaketal}.  Suppose $Z \subseteq V$ is a feedback vertex set of $G$. Then $G[V-Z]$ is a forest and therefore has at most $|V-Z|-1$ edges. An edge of $G$ is either in $G[V-Z]$ or is incident to a vertex of $Z$. Thus, $\sum_{u \in Z} d(u) + |V-Z| - 1 \ge |E|$ where $d(u)$ is the degree of $u$ in the graph; an edge between two vertices in $Z$ is double counted in this summation, and hence the inequality. Rearranging gives the inequality $\sum_{u \in Z} (d(u) - 1) \ge |E| - |V| + 1$.  Now suppose $\bar{x} \in \{0,1\}^V$ is an indicator vector for a feedback vertex set in the graph, then we can express the preceding inequality as
$$\sum_{u \in V} (d(u) - 1)x_u \ge |E| - |V| + 1.$$
Note that the summation is over all vertices since the terms for $u \in V-Z$ do not contribute because $\bar{x}_u =0$ for those vertices. Hence this is a valid inequality for an integer linear program for FVS. This valid inequality can be applied to any vertex induced subgraph $G[S]$ since a feedback vertex set $Z$ for $G$ induces a feedback vertex set $Z \cap S$ in $G[S]$. $P_{SD}(G)$ is the intersection of these valid inequalities when applied to every vertex induced subgraph of $G$. Now we explain the formulation $P_{\ESD}(G)$.  Instead of writing the valid inequality for only the vertex induced subgraphs of $G$ we can write it for every \emph{edge induced} subgraph of $G$. This is exactly $P_{\ESD}(G)$. Thus, $P_{\ESD}(G) \subseteq P_{SD}(G)$
and the inclusion is, in fact, strict---see Figure \ref{fig:esd-stronger-than-sd} for an example showing the strict inclusion. There are  advantages in working with
$P_{\ESD}(G)$ instead of $P_{SD}(G)$ that we will discuss later. It is perhaps a bit surprising that $P_{\ESD}(G)$ was not previously considered explicitly.   

\begin{figure}[!t]
    \centering
\includegraphics[width=0.4\textwidth]{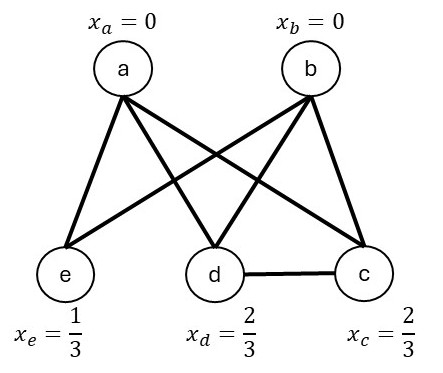}
\caption{A graph $G=(V, E)$ for which $P_{\ESD}(G)$ is strictly contained in $P_{\SD}(G)$. The point $x$ shown in the figure is in $P_{\SD}(G)$ but not in $P_{\ESD}(G)$. In particular, it violates the constraint for the edge-subset $F=E-\{\{c,d\}\}$.}
\label{fig:esd-stronger-than-sd} 
\end{figure}

The following are our main structural results. 
\begin{theorem}[Extreme Point Property of Strong Density Polyhedron]\label{thm:extreme-point}
Let $G$ be a graph containing a cycle and $x$ be an extreme point solution of $P_{SD}(G)$. Then, there exists $u \in V$ such that $x_u \ge 1/2$.
\end{theorem}

\begin{theorem}[Extreme Point Property of Edge Strong Density Polyhedron]\label{thm:esd-extreme-point}
Let $G$ be a graph containing a cycle and $x$ be an extreme point solution of $P_{\ESD}(G)$. Then, there exists $u \in V$ such that $x_u \ge 1/2$.
\end{theorem}

\paragraph{Solving the LP Relaxations and Iterative Rounding Algorithms.} Our structural result in Theorem \ref{thm:extreme-point} does not immediately imply iterative rounding algorithms since it is not obvious how to solve the LP relaxation $P_{SD}(G)$; no polynomial-time separation oracle is known. Nevertheless, we show that the extreme point results and additional structural properties of the constraint system can be leveraged, along with the Ellipsoid method, to obtain $2$-approximation algorithms via a round-or-cut approach embedded into iterative rounding --- see Appendix \ref{sec:rounding-algorithm}. 
On the other hand, we are able to derive a polynomial-time separation oracle for $P_{\ESD}(G)$ which, along with Ellipsoid and our extreme point result in Theorem \ref{thm:esd-extreme-point}, leads to a straight-forward iterative rounding based $2$-approximation --- see Section \ref{sec:esd-rounding-algorithm}. 
In another contribution, we show that a polynomial-sized extended formulation for FVS that was developed in \cite{CCFKW25}, based on orientation constraints that arise in densest subgraph problems, is an extended formulation of $P_{\ESD}(G)$. This implies that there is an efficient algorithm to optimize over $P_{\ESD}(G)$ while retaining the extreme point property (see Section \ref{sec:iter-orient}). One can use this extended formulation, along with the structural results, to obtain another iterative rounding $2$-approximation that does not rely on using the Ellipsoid method ---see Section \ref{sec:iter-orient}).
Note that we do not yet know an efficient separation oracle for $P_{SD}(G)$, which perhaps suggests that the stronger relaxation $P_{\ESD}(G)$ is more natural in a certain sense.

\paragraph{Motivations for the Extreme Point Conjecture.} 
The extreme point conjecture regarding strong density polyhedron for FVS \cite{Fiorini21,CCFKW25} was itself formulated as a stepping stone towards better approximations for two related generalizations:  Treewidth Deletion and Subset-FVS (SFVS). 
In the $\eta$-treewidth deletion problem, the input is a graph $G$ with non-negative vertex costs, and the goal is to remove a minimum-cost subset of vertices so that the remaining graph has treewidth at most $\eta$. 
A recent work showed a randomized constant-approximation for every fixed constant $\eta$ via combinatorial techniques \cite{wlodarczyk2025losing}; previously the unit-cost version of the problem admitted constant approximation \cite{fomin2012planar,GLLMW19}. Designing LP-based approximations for this problem remains open. See \cite{fomin2012planar,wlodarczyk2025losing} for important applications of this problem to a large class of vertex deletion problems and graph structure theory.
In Subset Feedback Vertex Set (SFVS), the input is a graph $G=(V,E)$ with non-negative vertex costs and a subset $S \subset V$ of terminals, and the goal is to remove a minimum-cost subset $X$ of vertices so that there is no cycle containing any terminal. 
The precise approximability of this problem is still undetermined, with the best-known lower bound being $2$ (coming from FVS) and the best-known upper bound being $8$ \cite{EvenNSZ00}.

In terms of techniques, most iterated rounding based algorithms follow extreme point properties for polyhedral relaxations that have $\{0,1\}$-coefficients in their constraint matrices \cite{LRS-book}. $P_{SD}(G)$ and $P_{\ESD}(G)$ do not fit this set up and provide interesting new examples where iterated rounding still works.

LP-based approximations via solvable LPs are of value in practice: they can be used to infer better instance-based approximation guarantees. In particular, we can solve the LP on the given instance and compare the optimum objective value of the LP to the solution generated by rounding algorithms to observe better approximation on the instance than the naive $2$-factor guarantee. 

\subsection{Technical Overview and AI Disclosure}\label{sec:techniques-and-ai-disclosure}
All graphs are finite, undirected, simple, and loopless unless explicitly stated otherwise. The proof of the structural result uses a high-level template that is inspired by Jain's seminal work on the Survivable Network Design Problem (SNDP) \cite{jain} and subsequent developments on iterated rounding proofs \cite{LRS-book}. However, there are several important differences and interesting challenges that we outline. We observe that the coefficients of the constraints describing both our polyhedra are not necessarily in $\{0, 1\}$ and could in fact be arbitrary non-negative integers. 
Almost all extreme point results in the literature have been shown only for polyhedra whose constraint coefficients are in $\{0, 1\}$. To the best of the authors' knowledge, the only exception is the extreme point result for pseudo-forest deletion shown in \cite{CCFKW25}. Despite this challenge, \cite{CCFKW25} were able to show and exploit the existence of a laminar basis for the extreme point, while the strong-density polyhedron $P_{\SD}(G)$ studied in this work does not admit a laminar basis.

We focus on $P_{\SD}(G)$. 
Let $x$ be a basic feasible solution. Under the assumption that there is no vertex $u$ with $x_u \ge 1/2$, we derive a contradiction. To obtain this contradiction, we show that there is a structured basis for $x$ (recall that a basis is a set of $n:=|V|$ tight constraints from the polyhedron that uniquely determine $x$).  In several iterated rounding proofs, especially in network design, one typically can assume that $x$ is fully fractional; however, that is not the case here. The set of vertices $Z = \{u \mid x_u = 0\}$ forms tight constraints in the basis. Recall that each row of the polyhedron corresponds to a subset $S \subseteq V$. Given $x$ and $S \subseteq V$ the
constraint corresponding to $S$ can be rewritten as $f_x(S) \le 0$ where $f_x(S) := \sum_{uv \in E[S]} (1-x_u-x_v) - \sum_{u \in S} (1-x_u) +1$.
An important observation was made in \cite{CCFKW25}: under the assumption that $x_u < 1/2$ for all $u \in V$, 
the set function $f_x$ is supermodular (this was termed conditional supermodularity). Using uncrossing techniques, we show that there is a basis of $x$ in which the tight non-trivial sets form an \emph{almost laminar family}: this is a family of sets in which every pair of sets satisfies either containment or intersect in at most one vertex. Our attempts to understand the structure of this almost-laminar family were unsuccessful, although we obtained several interesting technical observations. We used AI tools Claude (Opus 4.7) and Gemini (3.1 Pro) (via their desktop versions) that were fed our research notes consisting of our almost-laminar basis proof and other observations. Over multiple rounds of interaction, the tools helped unravel the proof. Two key ideas that the tools proposed were the following. If one considers the poset structure induced by sets in an almost-laminar family (via the usual set inclusion), then the Hasse diagram of this poset for an almost-laminar family is a forest if all sets are of size $\ge 2$; in our setting, the non-trivial sets in the basis are of size at least $3$ (since they need to contain a cycle). A second key observation was an elegant and non-trivial counting lemma that leads to a contradiction. Unlike well-known (fractional) token-counting arguments (e.g., see \cite{LRS-book}), the proof relies on a clever combinatorial inductive argument. For SNDP, a combinatorial inductive argument was given in \cite{chekuri2018note}. 

The AI-assisted proof for $P_{SD}(G)$ inspired us to consider the edge-density polyhedron $P_{\ESD}(G)$. We were aware of the edge-density polyhedron for some time, however our experiments with AI  focused primarily on $P_{\SD}(G)$. Note that the constraints of $P_{\ESD}(G)$ correspond to edge subsets instead of vertex subsets. In this setting, uncrossing yields a basis defined by a \emph{laminar} family of tight sets (indexed by edge subsets). Laminarity is simpler and more standard to deal with when compared to almost-laminarity; however, the vertex subsets induced by the edge subsets in the laminar family can intersect. We adapted the counting argument for $P_{SD}(G)$ to this setting to arrive at our second structural result.

Converting the structural results into an iterative rounding algorithm via the Ellipsoid method is via the, by-now standard, round-or-cut approach.
However, it is not straightforward and requires some care for $P_{SD}(G)$, and relies on the conditional supermodularity property that we mentioned above---see Appendix \ref{sec:rounding-algorithm}. Finally, we intuited the equivalence between $P_{\ESD}(G)$ and the extended formulation for FVS from \cite{CCFKW25} that was based on orientation variables. It was a natural conjecture that the orientation LP from \cite{CCFKW25} is equivalent to $P_{\ESD}(G)$ since a similar phenomenon was observed in \cite{CCFKW25} for a related problem called the pseudo-forest deletion problem (PFDS).

\paragraph{Declaration.} We used AI to write various parts of the proofs and rewrote them for clarity and readability. The authors assume full responsibility for all content. 

\paragraph{Discussion.} 
Laminarity has long served as a central tool in extreme point arguments for LPs with exponentially many constraints. 
Although we were able to prove the existence of an almost-laminar basis, the lack of laminarity made us uncertain about the conjecture. 
AI helped in identifying structure within the almost-laminar basis, which is quite simple in retrospect. Even with this observation, the counting argument is also elegant and may have taken time to figure out. 
The (dis)advantage of AI is that it does not pause for human collaborators.
We hope that the community will also benefit from the ideas in this proof.

\subsection{Related Work}
Vertex deletion to every non-trivial hereditary property is NP-Complete \cite{lewis1980node} of which FVS is a well-studied special case. There is extensive work on FVS and its generalizations and hence we limit our discussion here to closely related lines of work.

Fujito considered a matroidal generalization of FVS and showed that the primal-dual 
algorithm for FVS yields a $2$-approximation for a class of sparse-matroids \cite{Fuj-matroid-FVS}. This class includes the pseudo-forest deletion problem (PFDS), which is closely related to FVS; here we want to delete a min-cost subset of vertices such that the residual graph is a pseudo-forest\footnote{A pseudo-forest is a graph whose connected components are pseudo-trees (a tree plus one edge).}. 

Of particular relevance to this paper is the connection to the densest subgraph problem (DSG) and its deletion version. In DSG, the input is a graph $G=(V,E)$ and the goal is to find a subset $S \subseteq V$ that maximizes the edge density $|E(S)|/|S|$ where $E(S)$ is the set of edges with both endpoints in $S$. It is a well-known poly-time solvable problem (via reduction to network flows or to submodular function minimization). Charikar \cite{charikar_greedy_2000} gave an exact LP relaxation for this problem whose dual can be interpreted as a fractional orientation LP. FVS can be viewed as a density deletion problem: given a graph $G=(V,E)$, remove a min-cost subset of vertices such that the densest subgraph in the residual graph has density strictly less than $1$ (PFDS is the problem where we want density to be at most $1$). This connection to density 
inspired the extended formulations for FVS and PFDS in \cite{CCFKW25}, who used the density nomenclature. We continue to use the density nomenclature for the same reason. 
See \cite{ChandrasekaranCK2025} for the approximability of the density deletion problem when the density threshold is larger than $1$.

Subset Feedback Vertex Set (SFVS) is a generalization of FVS. The known lower bound on the approximability is $2$ (coming from FVS) and the best-known upper bound on the approximability is $8$ \cite{EvenNSZ00}. 
Chekuri and Madan \cite{chekuri-madan16} described an LP relaxation for SFVS and showed that its integrality gap is at most $13$. Recently \cite{CCFKW25} showed that the integrality gap of the formulation in \cite{chekuri-madan16} for the special case of FVS is $2$. 

FVS in directed graphs (DFVS) is also a well-studied problem: the goal is to remove a min-cost subset of vertices such that there are no directed cycles in the remaining graph. The best known approximation is $O(\log n \log \log n)$ via the hitting set LP \cite{Seymour1995,EvenNSS1998} and under UGC it is known that no constant factor is possible \cite{Svensson12b,GuruswamiLee2016}. FVS and DFVS have also been important problems in parameterized complexity. Both are known to be in FPT parameterized by the solution size \cite{cygan2015parameterized}.

As mentioned before, PFDS is closely related to FVS from the density deletion perspective. PFDS admits a $2$-approximation \cite{LFFW19} and does not admit a $(2-\epsilon)$-approximation for every fixed constant $\epsilon>0$ assuming UGC (via approximation preserving reduction from Vertex Cover). Chandrasekaran, Chekuri, Fiorini, Kulkarni, and Weltge \cite{CCFKW25} gave a weak density polyhedral formulation for PFDS that closely resembles the strong density polyhedron $P_{\SD}(G)$ for FVS. They showed an extreme point property for the weak density polyhedron for PFDS---namely, every extreme point has a coordinate with a value at least $1/3$. Both their polyhedron and $P_{\SD}(G)$ have constraint coefficients that are not necessarily in $\{0, 1\}$. However, their proof for the weak density polyhedron was able to show and exploit the existence of a laminar basis (for the extreme point all of whose coefficients are less than $1/3$), while our argument for $P_{\SD}(G)$ has to deal with an almost-laminar basis (for the extreme point all of whose coefficients are less than $1/2$). The almost-laminar basis necessitated substantially different arguments in this work.

\paragraph{Organization.} We prove our two structural results in Sections~\ref{sec:extreme-point} and \ref{sec:esd-extreme-point}. The two proofs are similar in several ways but have several technical differences. We have kept them self-contained so that a reader can read them independently in any order. The iterated rounding algorithm based on $P_{\ESD}(G)$ is simple and is included in Section \ref{sec:esd-extreme-point}. The corresponding algorithm based on $P_{SD}(G)$ is involved and is deferred to Appendix~\ref{sec:rounding-algorithm}. We discuss the Orientation polyhedron, its equivalence to $P_{\ESD}(G)$, and an iterative rounding via the Orientation polyhedron in Section~\ref{sec:ESD-and-orient-equivalence}. 
\section{Extreme Point Property of Strong Density Polyhedron}
\label{sec:extreme-point}
We prove Theorem \ref{thm:extreme-point} in this section. 
The proof proceeds by contradiction. We assume throughout that $x$ is an extreme point of $P_{SD}(G)$ satisfying $x_u < 1/2$ for all $u \in V$, and derive a contradiction.
Throughout we will let $n$ denote the number of vertices of the graph which is also the number of variables in the LP.

\paragraph{Notation.} Define functions $f_x, g_x, b : 2^V \to \mathbb{R}$, where for all $S\subseteq V$, we have 
\begin{align*}
f_x(S) &:= \sum_{uv \in E[S]} (1 - x_u - x_v) - \sum_{u \in S} (1 - x_u) + 1, \\
g_x(S) &:= \sum_{u \in S} (d_S(u) - 1) x_u, \text{ and}\\
b(S) &:= |E[S]| - |S| + 1. 
\end{align*}

The constraints of $P_{SD}(G)$ are of the form $g_x(S) \ge b(S)$ for all $S \subseteq V$ with $E[(S)] \neq \emptyset$. We observe that $f_x(S) = b(S) - g_x(S)$. Thus, the constraint $g_x(S) \ge b(S)$ is equivalent to $f_x(S) \le 0$.

\begin{proposition}\label{prop:sd-ineqs-equiv}
Let $x\in [0,1]^V$. Then, 
  $x\in P_{SD}(G)$ if and only if $f_x(S)\le 0$ for every $S\subseteq V$ with $E[S]\neq \emptyset$.
\end{proposition}

We say that a set $S$ with $E[S] \ne \emptyset$ is \emph{tight} (for a feasible $x$) if $g_x(S) = b(S)$.  Let $\mathcal{T} := \{ S \subseteq V : E[S] \ne \emptyset,\; g_x(S) = b(S) \}$ denote the family of tight sets.  For a subset $S\subseteq V$, define the row vector $\mathrm{row}(S)\in \mathbb{R}^V$ by \[ \mathrm{row}(S)_u := \begin{cases}
    d_S(u) - 1 \text{ if } u\in S, \\
    0 \text{ if } u\not\in S.
\end{cases}
\]
Let $Z := \{ u \in V : x_u = 0 \}$ denote the set of zero-coordinate vertices. For $S\subseteq V$, let $\mathrm{NZ}(S) := \{ u \in S : x_u > 0 \}$ denote the \emph{support} of $x$ in $S$ (equivalently the set of non-zero vertices in $S$).

Via standard polyhedral theory, an extreme point $x$ of $P_{SD}(G)$ is
the unique solution to a set of $|V|$ linearly independent tight
inequalities from the set of constraints. We call such a set of tight
inequalities a basis for $x$. There can be multiple bases that define
$x$, and later we will show the existence of a structured basis to
derive the desired contradiction. We will work with bases that include
all the tight constraints corresponding to $Z$ ($x_u = 0, u \in Z$).

First, we observe that the hitting set inequalities
for FVS are implied in $P_{SD}(G)$. 

\begin{lemma}\label{lem:cycle-constraints-implied}
  Let $C$ be a cycle of $G$. If $x \in P_{SD}(G)$ then $\sum_{u \in V(C)} x_u \ge 1$ holds.
\end{lemma}
\begin{proof}
  Consider $S = V(C)$ and the induced subgraph $G[S]$. Since $x \in P_{SD}(G)$,
  the inequality $g_x(S) \ge b(S)$ holds. Consider the case when $d_S(u) = 2$ for all $u \in S$ which means that
  the cycle $C$ has no chords. In this case $|E[S]| = |S|$ and hence, $b(S) = 1$. Therefore, $g_x(S) = \sum_{u \in S} x_u \ge b(S) = 1$ which is the desired inequality.
  Suppose $C$ has a chord. Then there is a cycle $C'$ with $V(C') \subset V(C)$ such that $C'$ does not have a chord. Hence the previous analysis applied to $C'$ implies that $\sum_{u \in V(C')} x_u \ge 1$ holds.
\end{proof}

We obtain the following corollary from Lemma \ref{lem:cycle-constraints-implied} since $x_u < 1/2$ for all $u\in V$.
\begin{corollary}
  \label{cor:three-nz}
  For every $S \subseteq V$ such that $G[S]$ contains a cycle, $|\mathrm{NZ}(S)| \ge 3$.
\end{corollary}

\subsection{Conditional Supermodularity and Uncrossing}

A key observation, that we borrow from \cite{CCFKW25}, is the conditional supermodularity of the function $f_x$.
  
\begin{lemma}[Conditional Supermodularity]\label{lem:supermod}
  Suppose $x_u < 1/2$ for all $u \in V$. Then, $f_x$ is a supermodular function.
\end{lemma}
\begin{proof}
Write $w_{uv} := 1 - x_u - x_v$ for each edge $uv \in E$. Since $x_u < 1/2$ for all $u$, we have $w_{uv} > 0$ for all edges. We can rewrite $f_x(S) = w(E[S]) - |S| + x(S) + 1$, where $w(E[S]) = \sum_{uv \in E[S]} w_{uv}$ and $x(S) = \sum_{u \in S} x_u$.
The function $w(E[S])$ is supermodular since $w_{uv} > 0$ for every edge $uv\in E$. The function $-|S| + x(S) + 1$ is modular. Therefore, $f_x$ is supermodular.
\end{proof}

We observe that a set $S\subseteq V$ with $E[S] \ne \emptyset$ is tight iff $f_x(S) = 0$. We have the following uncrossing lemma for tight sets. 

\begin{lemma}[Uncrossing Tight Sets]\label{lem:uncrossing}
Let $A, B\in \mathcal{T}$ with $|A\cap B|\ge 2$. Then, 
$E[A \cap B] \ne \emptyset$, both $A \cap B$ and $A \cup B$ are tight, and 
\[
\mathrm{row}(A) + \mathrm{row}(B) = \mathrm{row}(A \cap B) + \mathrm{row}(A \cup B).
\]
\end{lemma}

\begin{proof}
Since $A$ and $B$ are tight, $f_x(A) = f_x(B) = 0$.

For the sake of contradiction, suppose $|A \cap B| \ge 2$ with $E[A \cap B] = \emptyset$. 
We have that 
\begin{align*}
f_x(A \cap B) 
&= 0 - |A \cap B| + x(A \cap B) + 1 \quad \quad  \text{(since $E[A\cap B]=\emptyset$)}\\
&< 1 - \frac{|A \cap B|}{2} \quad \quad \text{(since $x_u < 1/2$ for all $u\in A\cap B$)}\\
&<  0. \quad \quad \text{(since $|A\cap B|\ge 2$)}
\end{align*}
Also $f_x(A \cup B) \le 0$ (by LP feasibility, since $A$ and $B$ both have edges, so $A \cup B$ has edges). 
By supermodularity, $0 = f_x(A) + f_x(B) \le f_x(A \cap B) + f_x(A \cup B) < 0$, a contradiction.
Therefore $E[A \cap B] \ne \emptyset$ whenever $|A \cap B| \ge 2$.

Next, we show that $A\cap B$ and $A\cup B$ are tight: we may assume that $|A\cap B|\ge 2$ and $E[A \cap B] \ne \emptyset$. 
Now $f_x(A \cap B) \le 0$ (by LP feasibility since $E[A\cap B]\neq \emptyset$) and $f_x(A \cup B) \le 0$ (by LP feasibility since $E[A\cup B]\neq \emptyset$). Supermodularity gives $0 = f_x(A) + f_x(B) \le f_x(A \cap B) + f_x(A \cup B) \le 0$, so equality holds throughout: $f_x(A \cap B) = f_x(A \cup B) = 0$, i.e., both $A\cap B$ and $A\cup B$ are tight.

Next, we prove the row identity. We first observe that equality in supermodularity implies there are no edges between $A \setminus B$ and $B \setminus A$. Indeed, recall that $w(E[S]) = \sum_{uv \in E[S]} w_{uv}$ where $w_{uv} = 1 - x_u - x_v > 0$. The supermodularity relies on the identity
\[
w(E[A]) + w(E[B]) = w(E[A \cap B]) + w(E[A \cup B]) - \sum_{uv \in \delta(A \setminus B,\, B \setminus A)} w_{uv},
\]
where $\delta(A \setminus B, B \setminus A)$ denotes edges with one endpoint in $A \setminus B$ and the other in $B \setminus A$. (Such edges are counted in $E[A \cup B]$ but in neither $E[A]$ nor $E[B]$, and they appear in neither $E[A \cap B]$.) Since we have equality $f_x(A) + f_x(B) = f_x(A \cap B) + f_x(A \cup B)$ and the modular terms cancel, we obtain $w(E[A]) + w(E[B]) = w(E[A \cap B]) + w(E[A \cup B])$. Since $w_{uv} > 0$ for all $uv \in E$, the sum over $\delta(A \setminus B, B \setminus A)$ must be zero, so $\delta(A \setminus B, B \setminus A) = \emptyset$.

Now we verify $\mathrm{row}(A)_u + \mathrm{row}(B)_u = \mathrm{row}(A \cap B)_u + \mathrm{row}(A \cup B)_u$ for each vertex $u$ by cases:
\begin{itemize}
\item $u \in A \cap B$: $d_A(u) + d_B(u) = d_{A \cap B}(u) + d_{A \cup B}(u)$, since every edge incident to $u$ within $A \cup B$ goes to $A \cap B$, $A \setminus B$, or $B \setminus A$. The edges to $A \cap B$ are counted on both sides. The edges to $A \setminus B$ contribute to $d_A(u)$ and $d_{A \cup B}(u)$; edges to $B \setminus A$ contribute to $d_B(u)$ and $d_{A \cup B}(u)$. Both sides sum to $d_{A \cap B}(u) + d_{A \setminus B}(u) + d_{B \setminus A}(u)$ where $d_X(u)$ counts edges from $u$ to $X$. Subtracting $1$ from each row: both sides give $d_A(u) - 1 + d_B(u) - 1 = d_{A \cap B}(u) - 1 + d_{A \cup B}(u) - 1$.
\item $u \in A \setminus B$: $\mathrm{row}(B)_u = \mathrm{row}(A \cap B)_u = 0$, so we need $\mathrm{row}(A)_u = \mathrm{row}(A \cup B)_u$, i.e., $d_A(u) = d_{A \cup B}(u)$. Since $\delta(A \setminus B, B \setminus A) = \emptyset$, vertex $u \in A \setminus B$ has no edges to $B \setminus A$, so $d_{A \cup B}(u) = d_A(u)$.
\item $u \in B \setminus A$: Symmetric to the previous case.
\item $u \notin A \cup B$: All four row entries are $0$.
\end{itemize}

\end{proof}

\subsection{Almost-Laminar Basis Structure}
The uncrossing lemma (Lemma \ref{lem:uncrossing}) allows us to show
the existence of a well-structured basis. We define the
well-structured property below and show the existence of such a basis
with additional properties in Theorem \ref{thm:basis}.

\begin{definition}[Almost-Laminar Family] A family $\mathcal{L}$ of subsets of $V$ is \emph{almost-laminar} if for every $A, B \in \mathcal{L}$, one of the following holds: (i) $A \subseteq B$, (ii) $B \subseteq A$, and (iii) $|A \cap B| \le 1$.
\end{definition}

We will work with row vectors coming from the constraints of the
polyhedron. Since there are $n$ variables these are $n$ dimensional
vectors.  For a set of (row) vectors $A$, we let $\subspan(A)$ to be
set of all (row) vectors spanned by the vectors in $A$. For a vertex
$u \in V$ we let $1_u$ denote the unit row vector correponding to $u$.
We let $\cZ = \{1_u \mid u \in Z\}$ denote the collection of the unit
row vectors corresponding to the vertices in $Z$.

First, we show some simple properties of non-singleton sets 
in a basis. 

\begin{lemma}
  \label{lem:basis-set-properties}
 Let $\cL$ be a basis for $x$ and let $\row(S) \in \cL$ where $|S|\ge 2$. Then $G[S]$ is connected and has a cycle.
\end{lemma}
\begin{proof}
  We argue in steps. First, there are no isolated vertices in $S$.
  If $v \in S$ is isolated in $G[S]$, then
  $g_x(S \setminus \{v\}) = g_x(S) + x_v = b(S) + x_v=b(S\setminus
  \{v\})-1 + x_v = b(S\setminus\{v\})-(1-x_v)<b(S\setminus \{v\})$,
  violating feasibility.

  Second, $G[S]$ is connected. Otherwise, $G[S]$ has connected
  components $S_1, S_2$ with $E[S_i] \ne \emptyset$ for both
  $i\in [2]$ (from the first property), then
  $g_x(S) = g_x(S_1) + g_x(S_2) \ge b(S_1) + b(S_2) = b(S) + 1 >
  b(S)$, contradicting tightness of $S$.

  Third, $G[S]$ contains a cycle. Otherwise, by the previous
  properties, $G[S]$ is a tree. Then $b(S) = 0$. Consequently,
  $g_x(S) = 0$ by tightness.  If there is a vertex $u \in S$ with
  $d_S(u) \ge 2$ and $x_u > 0$ then $g_x(S)$ has a strictly positive
  term $(d_S(u) - 1)x_u$ while $b(S) = 0$ which violates tightness
  of $S$. 
  Therefore, $x_u = 0$ for all $u$ with $d_S(u) \ge 2$. But
  this implies that $\row(S) \in \subspan(\cZ)$ contradicting linear
  independence of the vectors in the basis.
\end{proof}

\begin{theorem}[Almost-Laminar Basis]\label{thm:basis}
There exists an almost-laminar family $\mathcal{L} \subseteq \mathcal{T}$ such that:
\begin{enumerate}
\item 
the vectors $\{ \mathrm{row}(S) : S \in \mathcal{L} \} \cup \{ 1_u : u \in Z \}$ are linearly independent and $|\mathcal{L}| + |Z| = |V|$.
\item For every $S \in \mathcal{L}$, the graph $G[S]$ is $2$-connected 
and $|NZ(S)|\ge 3$.
\item If $A, B \in  \mathcal{L}$ and $A \subset B$, there exists $v \in B \setminus A$ with $x_v > 0$.
\end{enumerate}
\end{theorem}

\begin{proof}
We construct $\mathcal{L}$ in two stages: first uncrossing to an almost-laminar family, then refining to ensure $2$-connectivity. We will subsequently show that the third property also holds by exploiting the $2$-connectivity property. 

Fix a maximum cardinality $\cL\subseteq \cT$ such that 
\begin{enumerate}
    \item $\cL$ is almost-laminar and
    \item $\{\row(S): S\in \cL\}$ is linearly independent. 
    \end{enumerate}
    
    Let $\cW:=\subspan(\cL)$ and $\cQ:=\subspan(\cT)$.
    The following is the key lemma. 
\begin{lemma}\label{lem:span-tight-sets}
We have that $\cW=\cQ$. 
\end{lemma}
\begin{proof}
  Since $\cL\subseteq \cT$, it follows that $\cW\subseteq \cQ$. For the sake of contradiction, suppose $\cW\subsetneq \cQ$. Then, there exists $A\in \cT$ such that $\row(A)\not\in \cW$; among all such $A$, choose one that maximizes $|A|$.

  If $\cL\cup \{A\}$ is almost-laminar, then $\cL\cup \{A\}$
  contradicts the choice of $\cL$. Hence, $\cL\cup \{A\}$ is not
  almost laminar. Hence, there exists $B\in \cL$ such that
  $A-B\neq \emptyset$, $B-A\neq \emptyset$, and $|A\cap B|\ge 2$.
  Pick an inclusionwise minimal $B\in \cL$ that satisfies this
  property.

  Let $I:=A\cap B$ and $U:=A\cup B$. By Lemma \ref{lem:uncrossing}
  applied to $A, B\in \cT$, we have that
    \begin{enumerate}
        \item $I, U\in \cT$ and 
        \item $\row(A) + \row(B) = \row(I) + \row(U)$. 
    \end{enumerate}
    Since $B-A \neq \emptyset$, we have that $|U|>|A|$. If
    $\row(U)\not\in \cW$, then $U$ contradicts the choice of
    $A$. Therefore, $\row(U)\in \cW$. If $\row(I) \in \cW$ we would
    have $\row(A) \in \cW$ because
    $\row(A) + \row(B) = \row(I) + \row(U)$. However, by assumption on
    $\row(A) \not \in \cW$, and hence $\row(I) \not \in \cW$.

    Claim \ref{claim:almost-laminar} below
    shows that $\cL\cup \{I\}$ is almost-laminar.  Assuming the claim,
    $I\in \cT$,    $\row(I)\not\in \cW$, and $\cL\cup\{I\}$ is
    almost-laminar. Consequently, $\cL\cup \{I\}$ contradicts the
    choice of $\cL$, finishing the proof of the lemma.
    \end{proof}

    \begin{claim}\label{claim:almost-laminar}
        $\cL\cup\{I\}$ is almost-laminar. 
    \end{claim}
    \begin{proof}
        Let $T\in \cL$. It suffices to show that either $I\subseteq T$ or $T\subseteq I$ or $|I\cap T|\le 1$.
        Since $B, T\in \cL$, we have that 
        \begin{align*}
            &\text{either $B\subseteq T$ or $T\subseteq B$ or $|B\cap T|\le 1$.}
        \end{align*}
        We case based on the relationship between $B$ and $T$. 
        
        \noindent\textbf{Case 1.} $B\subseteq T$: Then, $I\subseteq B\subseteq T$ and consequently, $I\subseteq T$. 

        \noindent\textbf{Case 2.} $|B\cap T|\le 1$: Then, $I\cap T \subseteq B\cap T$ and consequently, $|I\cap T|\le 1$. 

        \noindent\textbf{Case 3.} $T\subseteq B$: 
        If $T=B$, then $I=B\cap A \subseteq B=T$ and we are done. Hence, we may assume that $T\subsetneq B$. 
        For the sake of contradiction, suppose $I-T\neq \emptyset$, $T-I\neq \emptyset$, and $|T\cap I|\ge 2$. Since $I= B\cap A$, we have that $A- T\supseteq I- T\neq \emptyset$, $T-A=T-I \neq \emptyset$, and $T\cap A = T\cap I$. Thus, $A-T\neq \emptyset$, $T-A\neq \emptyset$, and $|T\cap A|\ge 2$. Since $T\in \cL$ and $T\subsetneq B$, we have that $T$ contradicts the choice of $B$. 
      \end{proof}
        
        We pick $\cL_0$ to be an inclusion-wise maximal subfamily of
        $\cL$ such that $\{\row(S): S\in \cL_0\}\cup \cZ$
        are linearly independent. Recall that $\cZ = \{1_u \mid u \in
        Z\}$.  We note that
        $\subspan(\cL_0\cup \cZ)=\subspan(\cL\cup Z)=\subspan(\cT\cup
        Z)=\R^V$, where the first equality is because of the
        inclusion-wise maximal choice of $\cL_0$, the second equality
        is by Lemma \ref{lem:span-tight-sets}, and the third equality
        is because $\cT$ is the family of tight sets corresponding to
        the extreme point $x$. Thus, $\cL_0$ is an almost-laminar
        basis satisfying Property (1).

\medskip\noindent\textbf{Stage 2: Refinement to $2$-connected members.}
Starting from $\mathcal L_0$, we repeatedly perform the following operation.
If some member of $\cL_0$ is not $2$-connected, then choose 
an inclusionwise minimal 
set $S\in \cL_0$ that is not $2$-connected. 
By Lemma \ref{lem:basis-set-properties}, $G[S]$ is connected and has a cycle. 
Suppose $S \in \mathcal{L}_0$ has a cut vertex $v$ in $G[S]$. 
Choose a non-empty proper union $C_1$ of components of $G[S]-v$, and let $C_2$ be the union of the remaining components of $G[S]-v$. Define $S_i:=C_i\cup \{v\}$ for both $i\in [2]$. 
We observe that $g_x(S)
= g_x(S_1) + g_x(S_2) + x_v$  and $b(S) = b(S_1) + b(S_2)$.

Since $g_x(S) = b(S)$, $g_x(S_i) \ge b(S_i)$, and $x_v \ge 0$, we
conclude $x_v = 0$ and both $S_1, S_2 \in \mathcal{T}$. Moreover,
$\mathrm{row}(S) = \mathrm{row}(S_1) + \mathrm{row}(S_2) + 1_v$. Since
$x_v = 0$, we have that $1_v \in Z$.  Since
$\mathrm{row}(S) \notin \mathrm{span}(\mathcal{L}_0 \setminus \{S\}
\cup \cZ)$), at least one of $S_1$ and $S_2$ has
$\row(S_i)\not\in \mathrm{span}(\mathcal{L}_0 \setminus \{S\} \cup
\cZ)$. Without loss of generality, suppose it is $S_1$.
Then, we replace $S$ by $S_1$. This forms a new basis for $x$ since
the span of the collection is preserved. 

\emph{Almost-laminarity is preserved:} Consider an arbitrary
$T \in \mathcal{L}_0\setminus \{S\}$. If $T$ and $S$ were
incomparable, then $|T \cap S| \le 1$, so $|T \cap S_1| \le 1$. If
$T \supsetneq S$, then $T \supsetneq S_1$. Suppose $T \subsetneq S$. Then, by the inclusionwise minimal choice of $S$, the subgraph $G[T]$ is $2$-connected. Hence, $T$ lies entirely within $S_1$ or $S_2$ (as
$v$ separates $S_1 \setminus \{v\}$ from $S_2 \setminus \{v\}$ in
$G[S]$); if $T \subseteq S_1$ we have comparability, and if
$T \subseteq S_2$ then $T \cap S_1 \subseteq \{v\}$, giving
$|T \cap S_1| \le 1$.

This replacement reduces the total number of cut vertices across all
sets in the family. Repeating, we obtain $\mathcal{L}$ where every set
induces a $2$-connected graph. Let $S\in \mathcal{L}$.  By Lemma \ref{lem:basis-set-properties}, $G[S]$ contains a cycle. Since $G[S]$
is $2$-connected and contains a cycle, by Corollary
\ref{cor:three-nz}, we have $|\mathrm{NZ}(S)| \ge 3$.

This proves Properties (1) and (2).

\medskip\noindent\textbf{Property (3).}
Suppose $A \subsetneq B$ with $A, B \in \mathcal{L}$ and $x_v = 0$ for all $v \in B \setminus A$. We derive a contradiction.

Let $\ell := |\delta_G(A, B \setminus A)|$ denote the number of edges
between $A$ and $B \setminus A$. Since $G[B]$ is $2$-connected by
Property (2), we have $\ell \ge 2$. Since both $A$ and $B$ are tight and $x_v = 0$ for all $v \in B \setminus A$, we have that 
\begin{align*}
\frac{\ell}{2}
&> \sum_{uv \in \delta_G(A, B \setminus A)} x_u \quad \quad \text{(since $x_u<1/2$ for all $u\in A$)}\\
&= \sum_{u \in A}(d_B(u) - d_A(u))x_u\\ 
&= \sum_{u \in A}(d_B(u)-d_A(u))x_u + \sum_{u \in B-A} (d_B(u) - 1) x_u \quad \quad \text{(since $x_u=0$ for all $u\in B\setminus A$)}\\
&= \sum_{u \in B} (d_B(u)-1)x_u - \sum_{u \in A}(d_A(u) - 1)x_u \\
&=g_x(B) - g_x(A) \\
&=b(B) - b(A) \quad \quad \text{(since $A$ and $B$ are tight)}\\
& = |E[B]|-|B| - |E[A]|+|A| \\
& = \delta_G(A, B\setminus A) + |E[B \setminus A]| - |B \setminus A|\\
&= \ell + |E[B \setminus A]| - |B \setminus A|.
\end{align*}
Rearranging, we obtain that 
\begin{equation}\label{eq:ell-ub}
|E[B \setminus A]| < |B \setminus A| - \frac{\ell}{2}.
\end{equation}

On the other hand, since $G[B]$ is $2$-connected (by Property (2)), every vertex in $B$ has $d_B(v) \ge 2$. Thus, 
\[
2|E[B \setminus A]| + \ell = 
\sum_{v \in B \setminus A} d_B(v) \ge 2|B \setminus A|,
\]
giving
\begin{equation*}
|E[B \setminus A]| \ge |B \setminus A| - \frac{\ell}{2}, 
\end{equation*}
a contradiction to inequality \eqref{eq:ell-ub}. 

\end{proof}

\subsection{Forest Structure of Almost-Laminar Family}
Let $\mathcal{L}$ be the almost-laminar family from
Theorem~\ref{thm:basis}.  It naturally forms a partially ordered set
(poset) under set inclusion. We now analyze the structure of this
poset.  Consider the Hasse diagram of the poset
$(\mathcal{L}, \subseteq)$: recall that vertices of this diagram
correspond to sets in $\mathcal{L}$ and we have an edge from $A$ to
$B$ if $A$ is a minimal proper superset of $B$.  We will call two sets
$A$ and $B$ to be \emph{incomparable} if $A$ is not contained in $B$
and $B$ is not contained in $A$.

\begin{lemma}[Poset Structure]\label{lem:poset}
For every set $S\in \mathcal{L}$, there exists at most one minimal proper superset in $\mathcal{L}$. Consequently, the Hasse diagram of $(\mathcal{L}, \subseteq)$ is a forest (since every set has a unique parent in the Hasse diagram). 
\end{lemma}

\begin{proof}
The only obstruction to a forest Hasse diagram for an almost-laminar family is the
possibility that a singleton set has two incomparable parents. Here, every set contains at
least three support vertices, so this obstruction cannot occur. We give a formal proof. 
  For every $S\in \mathcal{L}$, $|S| \ge 3$ by Property (2) of
  Theorem~\ref{thm:basis}. For every pair of incomparable
  $A, B \in \mathcal{L}$, almost-laminarity gives $|A \cap B| \le 1$
  (since neither $A \subseteq B$ nor $B \subseteq A$).

Now suppose for contradiction, let $S \in \mathcal{L}$ have two distinct minimal proper supersets $P_1, P_2 \in \mathcal{L}$, i.e., $S \subsetneq P_1$, $S \subsetneq P_2$, and there is no set $B\in \mathcal{L}$ such that $S\subsetneq B \subsetneq P_1$ (resp.\ $P_2$). Since both are minimal proper supersets of $S$, neither $P_1 \subseteq P_2$ nor $P_2 \subseteq P_1$ (otherwise one would not be minimal). Thus, $P_1$ and $P_2$ are incomparable sets in $\mathcal{L}$ and hence, $|P_1\cap P_2|\le 1$. But $S \subseteq P_1 \cap P_2$, so $|P_1 \cap P_2| \ge |S| \ge 3$, a contradiction. 
\end{proof}

We note that the preceding lemma relies on the fact that $\mathcal{L}$
does not contain singleton sets. If singleton sets are allowed then
almost-laminarity does not suffice to obtain the forest structure.

Consider the Hasse diagram of the poset $(\mathcal{L}, \subseteq)$.
Since the Hasse diagram is a forest, each connected component is a
rooted tree (rooted at its unique maximal element, with edges directed
from parent to child in the containment order). A set may have
multiple children (sets for which it is the parent), but at most one
parent.  We define each connected component of this forest as a
\emph{block} and let $\numberofblocks$ denote the number of
blocks. Let $\mathcal{B}_1, \ldots, \mathcal{B}_\numberofblocks$
denote the blocks.  For each block $\mathcal{B}_i$, let
$k_i = |\mathcal{B}_i|$ be the number of sets it contains, and let
$\hat{V}_i \in \mathcal{B}_i$ be its unique maximal set (i.e., the
root of the tree). 

\begin{remark}[Inter-block structure]\label{rem:interblock}
Consider distinct $i, j\in [\numberofblocks]$. Then, for distinct blocks $\mathcal{B}_i, \mathcal{B}_j$, their roots $\hat{V}_i, \hat{V}_j$ are incomparable in $\mathcal{L}$, so $|\hat{V}_i \cap \hat{V}_j| \le 1$. More generally, every two sets from distinct blocks are incomparable: if $A \in \mathcal{B}_i$ and $B \in \mathcal{B}_j$ are comparable, say $A \subseteq B$, then $A \subseteq B \subseteq \hat{V}_j$, meaning $A$ would be comparable with $\hat{V}_j$ and hence in the same component as $\hat{V}_j$, contradicting $A \in \mathcal{B}_i$. Therefore, $|A \cap B| \le 1$ for every $A \in \mathcal{B}_i$, $B \in \mathcal{B}_j$. 
Since $|\hat{V}_i \cap \hat{V}_j| \le 1$, distinct blocks share no edges: $E[\hat{V}_i] \cap E[\hat{V}_j] = \emptyset$.
\end{remark}

\subsection{Counting Lemmas for Contradiction}
The goal is to derive a contradiction via the forest structure of the
Hasse diagram and a counting argument. We set up some basic notation.
From the basis structure, we have that $|\mathcal{L}|+|Z|=n$. We let
$z$ denote $|Z|$ and use $p = n-z$ to denote the number of vertices
$u$ with strictly positive $x_u$ value. We have
$p = |\cL| = \sum_{i=1}^t k_i$.

First consider the simple case when each block consists of a single
set. Thus $\cB_i = \{\hat{V}_i\}$. Suppose we further assume that
$\hat{V}_1,\ldots,\hat{V}_t$ are pairwise disjoint. Then a
contradiction is quite easy as follows. We have $p = |\cL| = t$ since
each block has a single set.  However, $|NZ(\hat{V}_i)| \ge 3$ for each $i$ by
Theorem~\ref{thm:basis} which implies that
$\sum_{i=1}^t |NZ(\hat{V}_i)| \ge 3t$. If the maximal sets are
pairwise disjoint then $p \ge \sum_{i=1}^t |NZ(\hat{V}_i)|$ which is a
contradiction since $t \ge 1$. 

However, the assumption
$\hat{V}_1,\ldots,\hat{V}_t$ are pairwise disjoint is too strong. We
only have almost-laminarity which implies that
$|\hat{V}_i \cap \hat{V_j}| \le 1$ for all $i \neq j$. Thus a vertex
$v$ may belong to multiple maximal sets and we can no longer obtain a
contradiction easily. We need to account for the overlap of the sets even in this restricted case.
This motivates the following definition. 

\begin{definition}[Global sharing loss]
For each vertex $v$ with $x_v > 0$, let 
\[m_v := |\{i \in [\numberofblocks] : v \in NZ(\hat{V}_i)\}|\] 
denote the number of blocks containing $v$. The \emph{global sharing loss} is
\[
\sigma_{\mathrm{global}} := \sum_{\substack{v \in V \\ x_v > 0}} (m_v - 1).
\]
\end{definition}

We observe that support vertices $v$ with $m_v = 1$ contribute zero in
the definition of $\sigma_{\mathrm{global}}$, so only support vertices
shared between multiple blocks contribute positive amount to
$\sigma_{\mathrm{global}}$.  
Moreover, every vertex $v\not\in Z$ lies in at least one root $\hat{V}i$ and hence, have $m_v\ge 1$; otherwise, all rows $row(S)$ for $S \in L$, and all unit rows $1_u$ for $u \in Z$, would have zero in coordinate $v$, contradicting that they form a basis of $\R^V$. 
The following proposition is easy from the definitions.

\begin{proposition}\label{prop:intermed}
We have that $\sum_{i=1}^\numberofblocks |NZ(\hat{V}_i)| = \nnz +
\sigma_{\mathrm{global}}$.
\end{proposition}
\begin{proof}
We observe that 
\begin{align*}
\sum_{i=1}^\numberofblocks |NZ(\hat{V}_i)| 
&= \sum_{i=1}^\numberofblocks \sum_{u\in V-Z} 1_{u\in \hat{V}_i} \\
&= \sum_{u\in V-Z}\sum_{i=1}^\numberofblocks 1_{u\in \hat{V}_i} = \sum_{u \in V-Z}m_v = \sum_{u\in V-Z}(m_v -1) + |V-Z| \\
& = \nnz + \sigma_{\mathrm{global}}. 
\end{align*}
\end{proof}

\paragraph{Two key lemmas.} The rest of the analysis is based on two
key technical lemmas. We state the lemmas and use them to complete the proof of Theorem \ref{thm:extreme-point} here. We will prove these two lemmas subsequently. The first bounds the sharing loss for a
collection of tight sets.

\begin{lemma}[Union Feasibility Inequality]\label{lem:feasibility}
Let $\mathcal{A} = \{A_1, \ldots, A_r\} \subseteq \mathcal{T}$ be a collection of tight sets such that $|A_j \cap A_\ell| \le 1$ for all distinct $j, \ell\in [r]$. 
Let $U=\cup_{j=1}^r A_j$ and for each $v\in U$, let 
\[
m_v^{\mathcal A} := |\{j\in [r]: v\in A_j\}|.
\]
Define the \emph{sharing loss} of $\mathcal{A}$ to be $\sigma(\mathcal{A}) := \sum_{v \in U} (m_v^{\mathcal A} - 1)$ (equivalently, $\sigma(\mathcal{A}) = \sum_{j=1}^r |A_j| - |U|$). 
If $r\ge 2$, then \[\sigma(\mathcal{A}) \le 2r - 3.\]
\end{lemma}
We will see later in the proof of Theorem \ref{thm:extreme-point} that we are interested in the sharing loss of support vertices only, but the above lemma is phrased in terms of the sharing loss of all vertices for notational ease. 
The second lemma builds on the preceding to show a block surplus property. 

\begin{lemma}[Block Surplus]\label{lem:block-surplus}
For each $i\in [\numberofblocks]$, we have that 
\[
|NZ(\hat{V}_i)| \ge k_i + 2.
\]
\end{lemma}

\begin{corollary}[Lower Bound on Global Sharing Loss]\label{cor:lower}
$\sigma_{\mathrm{global}} \ge 2\numberofblocks$.
\end{corollary}

\begin{proof}
Combining Proposition \ref{prop:intermed} 
and Lemma~\ref{lem:block-surplus} gives
\[
\nnz + \sigma_{\mathrm{global}} = \sum_{i=1}^\numberofblocks |NZ(\hat{V}_i)| \ge \sum_{i=1}^\numberofblocks (k_i + 2) = \nnz + 2\numberofblocks. \qedhere
\]
\end{proof}

We postpone the proofs of these lemmas and obtain the desired
contradiction first.

\subsubsection{Proof of Theorem \ref{thm:extreme-point}}

\begin{proof}[Proof of Theorem~\ref{thm:extreme-point}]
Assume for contradiction that $x_u < 1/2$ for all $u \in V$.

Since $G$ contains a cycle $C$, Corollary~\ref{cor:three-nz} implies that
$\nnz \ge 3$, and hence $\numberofblocks \ge 1$. 

By Corollary~\ref{cor:lower}, we have that $\sigma_{\mathrm{global}} \ge 2\numberofblocks$.

If $\numberofblocks = 1$, then we have only one block and consequently no shared vertices across blocks, and hence,  $\sigma_{\mathrm{global}} = 0$, a contradiction. Thus, we may assume that $\numberofblocks \ge 2$. The block roots $\hat{V}_1, \ldots, \hat{V}_\numberofblocks$ are incomparable, so $|\hat{V}_i \cap \hat{V}_j| \le 1$ for every distinct $i, j\in [\numberofblocks]$. 
By Lemma~\ref{lem:feasibility} applied to the collection $\mathcal{A} = \{\hat{V}_1, \ldots, \hat{V}_\numberofblocks\}$, the total sharing loss $\sigma(\mathcal{A})$ among the roots satisfies $\sigma(\mathcal{A}) \le 2\numberofblocks - 3$.
Since $\sigma_{\mathrm{global}}$ counts the sharing loss of only support vertices while $\sigma(\mathcal{A})$ counts the sharing loss of all vertices in $\cup_{i\in [\numberofblocks]}\hat{V}_i$, 
we have $\sigma_{\mathrm{global}} \le \sigma(\mathcal{A}) \le 2\numberofblocks - 3$, a contradiction. 

\end{proof}

\subsubsection{Proof of the Union Feasibility Inequality}
\label{sec:ufi}
The intuition for the inequality is the following: When equations for tight sets are summed up, the right-hand side $b(·)$ gains a surplus equal to the number of vertex identifications, namely the sharing loss $\sigma$, minus the number of merged components. Feasibility of the union forces the fractional mass on shared vertices to pay for this surplus. Since each coordinate is below $1/2$, a shared vertex cannot pay one full unit, which bounds the total amount of sharing.
\begin{proof}[Proof of Lemma \ref{lem:feasibility}]
We first show the following inequality which will be useful to prove the lemma. 
\[
\sum_{v \in U} (m_v^{\mathcal A} - 1)\, x_v \;\ge\; \sigma(\mathcal{A}) - r + 1.
\]
  Because $|A_j \cap A_\ell| \le 1$, the sets share no edges, i.e.,
  $E[A_j] \cap E[A_\ell] = \emptyset$ for every distinct
  $j, \ell\in [r]$.  Let $F := E[U]\setminus (\cup_{j=1}^r E[A_j])$
  (essentially, $F$ is the set of crossing edges). The sets
  $E[A_1], \ldots, E[A_r], F$ partition $E[U]$. 

For each $u \in U$, we have $d_{U}(u) = \sum_{j\in [r]: u \in A_j}
d_{A_j}(u) + |F \cap \delta(u)|$. Therefore, 
\begin{align*}
g_x(U) 
&= \sum_{u \in U} d_U(u)\, x_u - x(U) 
= \sum_{j=1}^r \sum_{uv \in E[A_j]} (x_u + x_v) + \sum_{uv \in F} (x_u + x_v) - x(U).
\end{align*}
Using $x(U) = \sum_{j=1}^r x(A_j) - \sum_{v \in U} (m_v^{\mathcal A} - 1) x_v$ and $\sum_{uv\in E[A_j]}(x_u + x_v) = g_x(A_j) + x(A_j)$, we get that 
\begin{equation}\label{eq:gx}
g_x(U) = \sum_{j=1}^r g_x(A_j) + \sum_{v \in U} (m_v^{\mathcal A} - 1)\, x_v + \sum_{uv \in F} (x_u + x_v).
\end{equation}

Furthermore, 
\begin{align*}
|E[U]| &= \sum_{j=1}^r |E[A_j]| + |F| \text{ and}\\
|U| &= \sum_{j=1}^r |A_j| - \sigma(\mathcal{A}).
\end{align*}
Therefore, 
\begin{equation}\label{eq:bV}
b(U) = |E[U]| - |U| + 1 = \sum_{j=1}^r \big(|E[A_j]| - |A_j| + 1\big) + \sigma(\mathcal{A}) - r + 1 + |F|.
\end{equation}

Since each $A_j$ is tight, $g_x(A_j) = b(A_j)$. Subtracting \eqref{eq:bV} from \eqref{eq:gx} gives
\[
g_x(U) - b(U) = \sum_{v \in U} (m_v^{\mathcal A} - 1)\, x_v - (\sigma(\mathcal{A}) - r + 1) - \sum_{uv \in F} (1 - x_u - x_v).
\]
Since $x_u < 1/2$ for all $u\in V$, we have that $\sum_{uv\in
  F}(1-x_u-x_v)\ge 0$ (could be equal to zero if $F = \emptyset$). By LP feasibility, since $U$ has edges (as each $A_j$ does), $g_x(U) - b(U) \ge 0$. Rearranging yields
\[
\sum_{v \in U} (m_v^{\mathcal A} - 1)\, x_v \ge \sigma(\mathcal{A}) - r + 1 + \sum_{uv\in F} (1-x_u-x_v) \ge \sigma(\mathcal{A}) - r + 1.
\]

Suppose $r\ge 2$. If $\sigma(\mathcal A)=0$, the bound is immediate.  Otherwise,
using $x_v<1/2$ for all $v$ gives
\[
\frac{\sigma(\mathcal{A})}{2} > \sum_{v \in U} (m_v^{\mathcal A} - 1)\, x_v \ge \sigma(\mathcal{A}) - r + 1 \implies \frac{\sigma(\mathcal{A})}{2} < r - 1 \implies \sigma(\mathcal{A}) < 2r - 2.
\]
Since $\sigma(\mathcal{A})$ is an integer, we have that $\sigma(\mathcal{A}) \le 2r - 3$.
\end{proof}

\subsubsection{Proof of the Block Surplus Lemma}

\begin{proof}[Proof of Lemma \ref{lem:block-surplus}]
  We prove by induction on the tree structure of the block that for
  every $S \in \mathcal{B}_i$, if
  $\mathcal{T}_S \subseteq \mathcal{B}_i$ is the subtree of sets
  contained in $S$, and $k_S = |\mathcal{T}_S|$, then
  $|\mathrm{NZ}(S)| \ge k_S + 2$. Applying this to the root
  $S = \hat{V}_i$ yields the lemma.

\textbf{Base case ($k_S = 1$):} $S$ is a leaf in the containment tree. By Theorem~\ref{thm:basis}(2), 
$|\mathrm{NZ}(S)| \ge 3 = k_S + 2$.

\textbf{Inductive step:} Suppose $S$ has children $C_1, \ldots, C_r$ in $\mathcal{L}$. By induction, $|\mathrm{NZ}(C_j)| \ge k_{C_j} + 2$ for each child. We note that $k_S = 1 + \sum_{j=1}^r k_{C_j}$: the subtree $\mathcal{T}_S$ consists of $S$ itself together with the subtrees $\mathcal{T}_{C_1}, \ldots, \mathcal{T}_{C_r}$, which are disjoint (since the children are in the same block and the Hasse diagram is a forest, no set in $\mathcal{T}_{C_j}$ belongs to $\mathcal{T}_{C_\ell}$ for $j \ne \ell$).

Because $C_1, \ldots, C_r$ are incomparable in $\mathcal{L}$, we have that $|C_j \cap C_\ell| \le 1$ for every distinct $j, \ell\in [r]$. Let $U := \bigcup_{j=1}^r C_j$ and $\sigma_{\mathrm{NZ}}:=\sum_{j=1}^r |C_j\cap (V-Z)| - |U\cap (V-Z)|$. We note that $\sigma_{\mathrm{NZ}}$ is the sharing loss of support vertices in $U$ among the children. 
The size of the support of $U$ is:
\[
|\mathrm{NZ}(U)| = \sum_{j=1}^r |\mathrm{NZ}(C_j)| - \sigma_{\mathrm{NZ}},
\]
where $\sigma_{\mathrm{NZ}}$ is the sharing loss of support vertices among the children. 

\textbf{Case $r = 1$:} $U = C_1$, so $|\mathrm{NZ}(U)| = |\mathrm{NZ}(C_1)| \ge k_{C_1} + 2$. By Theorem~\ref{thm:basis}(3), since $C_1 \subsetneq S$, there is a vertex $v \in S \setminus C_1$ with $x_v > 0$. Therefore, $|\mathrm{NZ}(S)| \ge |\mathrm{NZ}(U)| + 1 \ge k_{C_1} + 3 = k_S + 2$.

\textbf{Case $r \ge 2$:} By Lemma~\ref{lem:feasibility} applied to the collection $\mathcal{A}:=\{C_1, \ldots, C_r\}$, the total sharing loss among the children is at most $\sigma(\mathcal{A}) \le 2r - 3$. We also have that $\sigma_{\mathrm{NZ}} \le \sigma(\mathcal{A})$ since $\sigma_{\mathrm{NZ}}$ counts the sharing loss of only support vertices in $U$ among the children while $\sigma(\mathcal{A})$ counts the sharing loss of all vertices in $U$ among the children. Hence, we have $\sigma_{\mathrm{NZ}} \le 2r - 3$.
Therefore:
\begin{align*}
|\mathrm{NZ}(S)| 
&\ge |\mathrm{NZ}(U)| 
= \sum_{j=1}^r |\mathrm{NZ}(C_j)| - \sigma_{\mathrm{NZ}}
\ge \sum_{j=1}^r (k_{C_j} + 2) - (2r - 3) \\
&= \sum_{j=1}^r k_{C_j} + 2r - 2r + 3 
= \sum_{j=1}^r k_{C_j} + 3 = k_S + 2.
\end{align*}
\end{proof}
\section{Extreme Point Property of Edge Strong Density Polyhedron}\label{sec:esd-extreme-point}
We prove Theorem \ref{thm:esd-extreme-point} in this section and use it to design a $2$-approximation via iterative rounding in Section \ref{sec:esd-rounding-algorithm}. 
The proof of Theorem \ref{thm:esd-extreme-point} proceeds by contradiction. We assume throughout that $x$ is an extreme point of $P_{\ESD}(G)$ satisfying $x_u < 1/2$ for all $u \in V$, and derive a contradiction. 
We recall that $P_{\ESD}(G)$ is defined by constraints on edge subsets (in contrast to $P_{\SD}(G)$, which is defined by constraints on vertex subsets). We also recall that for an edge-subset $F\subseteq E$, the vertex-subset $V(F)$ denotes the subset of vertices incident to edges in $F$, the subgraph $G[F]$ denotes $(V(F), F)$, and $d_F(u)$ denotes the degree of vertex $u$ in the subgraph $G[F]$. The overall proof is similar to that of Theorem \ref{thm:extreme-point}, but there are some technical differences. We have opted not to compress the proof so that readers can read the two sections independently. 

\paragraph{Notation.}
Throughout, we use $n$ to denote the number of vertices of $G$. 
Define functions $f_x, g_x, b: 2^E\rightarrow \R$, where for all $F\subseteq E$, we have 
\begin{align*}
 f_x(F)&:=\sum_{uv\in F}(1-x_u-x_v)-\sum_{u\in V(F)}(1-x_u)+1,\\
 g_x(F)&:=\sum_{u\in V(F)}(d_F(u)-1)x_u, \text{ and}\\
 b(F)&:=|F|-|V(F)|+1.
\end{align*}
The constraints of $P_{\ESD}(G)$ are of the form $g_x(F)\ge b(F)$ for all non-empty subset $F\subseteq E$. We observe that $f_x(F) = b(F) -g_x(F)$. Thus, the constraint $g_x(F)\ge b(F)$ is equivalent to $f_x(F)\le 0$. 

\begin{proposition}
Let $x\in [0,1]^V$. Then, 
    $x\in P_{\ESD}(G)$ if and only if $f_x(F)\le 0$ for every non-empty subset $F\subseteq E$. 
\end{proposition}

We say that a nonempty edge subset $F\subseteq E$ is \emph{tight} if $g_x(F)=b(F)$.  Let $\mathcal T:=\{\emptyset\ne F\subseteq E: g_x(F)=b(F)\}$ denote the family of tight edge sets.  For an edge subset $F\subseteq E$, define the row vector
$\row(F)\in \mathbb R^V$ by
\[
        \row(F)_u:=
        \begin{cases}
        d_F(u)-1, & u\in V(F),\\
        0, & u\notin V(F).
        \end{cases}
\]
Let $Z:=\{u\in V:x_u=0\}$ denote the set of zero-coordinate vertices. For $F\subseteq E$, let $\NZ(F):=\{u\in V(F):x_u>0\}$ denote the \emph{support} of $x$ in $V(F)$ (equivalently, the set of non-zero vertices in $V(F)$). 

Via standard polyhedral theory, an extreme point $x$ of $P_{\ESD}(G)$ is the unique solution to a set of $|V|$ linearly independent tight inequalities from the set of constraints. We call such a set of tight inequalities a basis for $x$. There can be multiple bases that define $x$, and later we will show the existence of a structured basis to derive the desired contradiction. We will work with bases that include all the tight constraints corresponding to $Z$ (namely, $x_u=0$ for $u\in Z$). 

\begin{lemma}[Large Support in Cyclic Subgraphs]\label{lem:esd-three-nz}
For every $F \subseteq E$ such that $G[F]$ contains a cycle, we have that $|\NZ(F)| \ge 3$.
\end{lemma}
\begin{proof}
    Let $C\subseteq F$ be the edge set of a simple cycle.  For this set, every vertex
of $V(C)$ has degree two in $C$, and hence the constraint corresponding to $C$ is
\[
        \sum_{u\in V(C)} x_u \ge 1.
\]
Since every coordinate is strictly smaller than $1/2$, at least three vertices
of $C$ have positive $x$-value.  Such vertices also belong to $V(F)$.
\end{proof}

\subsection{Supermodularity and Uncrossing}
We start with the supermodularity property of $f_x$ and a consequence of the tightness of the supermodularity inequality. The results of this section do not need $x_u < 1/2$ for all $u\in V$, unlike in the setting of $P_{SD}(G)$. They do need $x_u<1$ for all $u\in V$. 
\begin{lemma}[Supermodularity]\label{lem:esd-supermod}
Let $x\in [0, 1)^V$. Then, 
\begin{enumerate}
    \item $f_x$ is a supermodular function, and 
    \item for non-empty $A, B\subseteq E$, we have that $f_x(A\cap B) + f_x(A\cup B) =f_x(A) + f_x(B)$ if and only if $V(A) \cap V(B) =V(A\cap B)$. 
\end{enumerate}
\end{lemma}
\begin{proof}
Let $A, B\subseteq E$. Then, 
\[
 f_x(A\cap B)+f_x(A\cup B)-f_x(A)-f_x(B)
 =\sum_{u\in (V(A)\cap V(B))\setminus V(A\cap B)}(1-x_u)\ge 0.
\]
The last inequality is because $x_u\le 1$ for every $u\in V$. This proves supermodularity of $f_x$ and moreover, $f_x(A\cap B) + f_x(A\cup B) =f_x(A) + f_x(B)$ if and only if $V(A) \cap V(B) =V(A\cap B)$ (since $x_u<1$ for all $u\in V$). 
\end{proof}

We observe that a non-empty subset $F\subseteq E$ is tight iff $f_x(F) = 0$. We have the following uncrossing lemma for tight sets. 

\begin{lemma}[Uncrossing tight edge sets] \label{lem:esd-uncrossing}
Let $A,B\in\mathcal T$ with $A\cap B\ne\emptyset$.  Then,
$A\cap B$ and $A\cup B$ are tight and 
\[
        \row(A)+\row(B)=\row(A\cap B)+\row(A\cup B).
\]
\end{lemma}
\begin{proof}
    Since $A\cap B$ and $A\cup B$ are nonempty, feasibility
gives $f_x(A\cap B)\le 0$ and $f_x(A\cup B)\le 0$.  Since $A$ and $B$ are tight,
$f_x(A)=f_x(B)=0$.  By supermodularity of $f_x$ (as shown in Lemma \ref{lem:esd-supermod}),
\[
        0=f_x(A)+f_x(B)\le f_x(A\cap B)+f_x(A\cup B)\le 0.
\]
Thus equality holds throughout, and $A\cap B,A\cup B\in\mathcal T$.

It remains to prove the row identity.  By Lemma \ref{lem:esd-supermod}, we have that 
$V(A)\cap V(B)=V(A\cap B)$.  For each vertex $u$, let
$\eta_F(u):=1$ if $u\in V(F)$ and $\eta_F(u):=0$ otherwise.  Since edge degrees
are modular with respect to union and intersection,
\[
        d_A(u)+d_B(u)=d_{A\cap B}(u)+d_{A\cup B}(u) \qquad \forall u\in V.
\]
Also $\eta_A(u)+\eta_B(u)=\eta_{A\cap B}(u)+\eta_{A\cup B}(u)$ for all $u\in V$ since $V(A)\cap V(B)=V(A\cap B)$. Hence, $\row(A)_u + \row(B)_u = d_A(u) + d_B(u) - \eta_A(u) - \eta_B(u) = d_{A\cap B}(u)+d_{A\cup B}(u) - \eta_{A\cap B}(u)-\eta_{A\cup B}(u) = \row(A\cap B)_u+\row(A\cup B)_u$ for all $u\in V$. 

\end{proof}

\subsection{Laminar Basis Structure}
The uncrossing lemma (Lemma \ref{lem:esd-uncrossing}) allows us to show the existence of a laminar basis.

\begin{definition}[Edge-laminar family]
A family $\mathcal L\subseteq 2^E$ is \emph{laminar} if for every
$A,B\in\mathcal L$, one of the following holds: $A\subseteq B$ or $B\subseteq A$
or $A\cap B=\emptyset$.
\end{definition}

We emphasize that for sets $A, B\in \mathcal{L}$ for a laminar family $\mathcal{L}\subseteq 2^E$, we need not necessarily have that $V(A)\subseteq V(B)$ or $V(B)\subseteq V(A)$ or $V(A)\cap V(B)=\emptyset$. 

We will work with row vectors coming from the constraints of the polyhedron. Since there are $n$ variables, these are $n$-dimensional vectors. For a set of (row) vectors $A$, we let span$(A)$ be the set of all (row) vectors spanned by the vectors in $A$. For a vertex $u\in V$, we let $1_u$ denote the unit row vector corresponding to $u$.
We let $\cZ = \{1_u \mid u \in Z\}$ denote the collection of the unit row vectors corresponding to the vertices in $Z$. 

First, we show some simple properties of non-singleton sets in a basis. 
\begin{lemma}\label{lem:esd-simple-properties}
    Let $\mathcal{L}$ be a basis for $x$ and let row$(F)\in \mathcal{L}$ for some non-empty subset $F\subseteq E$. Then, $G[F]=(V(F), F)$ is connected and has a cycle. 
\end{lemma}
\begin{proof}
Let the edge-sets of the connected components of $G[F]$ be $F_1,\ldots,F_k$ with $k\ge 2$. These components are vertex-disjoint and
\[
        g_x(F)=\sum_{i=1}^k g_x(F_i)
        \ge \sum_{i=1}^k b(F_i)
        = b(F)+k-1>b(F),
\]
contradicting tightness. Thus, $G[F]$ is connected.  If $G[F]$ is a tree, then
$b(F)=0$, and tightness gives $g_x(F)=0$.  Since all coefficients
$d_F(u)-1$ are nonnegative, every vertex with $d_F(u)\ge 2$ has $x_u=0$.
Consequently, $\row(F)$ lies in $\operatorname{span}\{1_u:u\in Z\}$, contradicting the linear independence of the vectors in the basis.
\end{proof}

\begin{theorem}[Laminar tight basis]\label{thm:esd-basis}
There exists a laminar family $\mathcal L\subseteq\mathcal T$ such that:
\begin{enumerate}
    \item The vectors $\{\row(F):F\in\mathcal L\}\cup\{1_u:u\in Z\}$ are
    linearly independent and $|\mathcal L|+|Z|=|V|$.
    \item For every $F\in\mathcal L$, the graph $G[F]$ is 2-connected and
    $|\NZ(F)|\ge 3$.
    \item For every $A,B\in\mathcal L$ with $A\subsetneq B$, there exists
    $v\in V(B)\setminus V(A)$ with $x_v>0$.
\end{enumerate}
\end{theorem}
\begin{proof}
    We construct $\cL$ in two stages: we first uncross to a laminar basis, and then refine the members so
that they are 2-connected. We will subsequently show that the third property also holds by exploiting the $2$-connectivity property. 

\paragraph{Stage 1: uncrossing to a laminar family.}
Choose a laminar family $\mathcal L\subseteq\mathcal T$ of maximum
cardinality subject to the rows $\{\row(F):F\in\mathcal L\}$ being linearly
independent.  Let $W:=\operatorname{span}\{\row(F):F\in\mathcal L\}$ and $Q:=\operatorname{span}\{\row(F):F\in\mathcal T\}$. The following is the key lemma. 

\begin{lemma}\label{lem:esd-span-tight-sets}
We have that $\cW=\cQ$. 
\end{lemma}
\begin{proof}
Since $\cL\subseteq \mathcal{T}$, it follows that $\cW\subseteq \cQ$. For the sake of contradiction, suppose $\cW\subsetneq \cQ$. Then, there exists $A\in \mathcal{T}$ such that $\row(A)\not\in \cW$. 
\begin{align}
        &\text{Pick a max-sized $A\in \cT$ such that $\row(A)\not\in \cW$.}
\end{align}

If $\mathcal L\cup\{A\}$ is laminar, then
$\mathcal L\cup\{A\}$ contradicts the maximality of $\mathcal L$.  Hence there
exists $B\in\mathcal L$ such that $A\cap B\ne\emptyset$, $A\setminus B\ne
\emptyset$, and $B\setminus A\ne\emptyset$.  
\begin{align}
        &\text{Pick an inclusionwise minimal $B\in \cL$ such that $A\cap B\ne\emptyset$, $A\setminus B\ne \emptyset$, and $B\setminus A\ne\emptyset$.}
\end{align}
    
Let $I:=A\cap B$ and $U:=A\cup B$.  By Lemma \ref{lem:esd-uncrossing}, both $I$ and $U$ are
tight and
\[
        \row(A)+\row(B)=\row(I)+\row(U).
\]
Since $B\setminus A\ne\emptyset$, we have $|U|>|A|$, and therefore
$\row(U)\in \cW$ by the maximal choice of $A$.  As $\row(B)\in \cW$ and
$\row(A)\notin \cW$, the row identity implies $\row(I)\notin \cW$.

We next show that $\mathcal L\cup\{I\}$ is laminar.  Let $T\in\mathcal L$.
Since $B,T\in\mathcal L$, either $B\subseteq T$, or $T\subseteq B$, or
$B\cap T=\emptyset$.  If $B\subseteq T$, then $I\subseteq T$.  If
$B\cap T=\emptyset$, then $I\cap T=\emptyset$.  Finally suppose $T\subseteq B$.
If $T=B$, then $I\subseteq T$.  Otherwise $T\subsetneq B$.  If $I$ and $T$ were
not laminar, then $A$ and $T$ would also not be laminar: indeed,
$I\cap T\ne\emptyset$, $I\setminus T\ne\emptyset$, and $T\setminus I\ne
\emptyset$ imply $A\cap T\ne\emptyset$, $A\setminus T\ne\emptyset$, and
$T\setminus A\ne\emptyset$.  Thus, $T$ contradicts the minimal choice of $B$.
Thus $\mathcal L\cup\{I\}$ is laminar.  Thus, $I\in \cT$, $\row(I)\notin \cW$, and $\cL\cup \{I\}$ is laminar. Consequently, $\cL\cup \{I\}$ contradicts the choice of $\cL$. Therefore $\cW=\cQ$.
\end{proof}

Now choose an inclusionwise maximal subfamily $\mathcal L_0\subseteq\mathcal L$
such that
\[
        \{\row(F):F\in\mathcal L_0\}\cup\{1_u:u\in Z\}
\]
are linearly independent.  
We note that $\subspan(\cL_0\cup Z)=\subspan(\cL\cup Z)=\subspan(\cT\cup Z)=\R^V$, where the first equality is because of the inclusion-wise maximal choice of $\cL_0$, the second equality is by Lemma \ref{lem:esd-span-tight-sets}, and the third equality is because $\cT$ is the family of tight sets corresponding to the extreme point $x$. Thus, $\cL_0$ is a laminar basis satisfying Property (1). 

\paragraph{Stage 2: Refinement to $2$-connected members.}
Starting from $\mathcal L_0$, we repeatedly perform the following operation.
If some member of $\mathcal{L}_0$ is not $2$-connected, then choose 
an inclusionwise minimal such member $F$. 
By Lemma \ref{lem:esd-simple-properties}, $G[F]$ is connected and has a cycle. 
Let $v$ be a cut vertex of $G[F]$.  
Let $Q_1$ be a connected component of $G[F]-v$ and $Q_2$ be the remaining components of $G[F]-v$ that are not $Q_1$. For each $i\in [2]$, let $F_i$ be the union of the edge-set of $Q_i$ and the set of edges between $v$ and the vertices of $Q_i$. 
We note that $F_1$ and $F_2$ are non-empty. 
Then, $F=F_1\dot\cup F_2$  and $V(F_1)\cap V(F_2)=\{v\}$.
We observe that 
\begin{align*}
        g_x(F)&=g_x(F_1)+g_x(F_2)+x_v \text{ and}\\
        b(F)&=b(F_1)+b(F_2).
\end{align*}
Since $g_x(F)=b(F)$ and $g_x(F_i)\ge b(F_i)$ for both $i\in [2]$, we obtain $x_v=0$ and
$F_1,F_2\in\mathcal T$.  Moreover,
\[
        \row(F)=\row(F_1)+\row(F_2)+1_v.
\]
Since $v\in Z$, at least one
of $\row(F_1),\row(F_2)$ is 
not in $\mathrm{span}(\cL_0 \setminus \{F\}\cup Z)$. 
Replace $F$ by such an $F_i$ in the basis. This replacement preserves the basis property.

The replacement also preserves
laminarity:  Indeed, any set outside $F$ either contains $F$ or is
disjoint from $F$, so it either contains the chosen $F_i$ or is disjoint from it.
Any proper set $T\subsetneq F$ in the family is 2-connected by the minimal choice
of $F$.  Since $F_1$ and $F_2$ meet only at the cut vertex $v$, such a
2-connected $T$ cannot use edges from both $F_1$ and $F_2$; otherwise $v$ would
be a cut vertex of $G[T]$.  Hence, $T$ is either
contained in the chosen $F_i$ or is edge-disjoint from it.

This replacement reduces the total number of edges across all sets in the family. Repeating, we obtain $\mathcal{L}$ where every set induces a $2$-connected graph. Let $F\in \mathcal{L}$. By Lemma \ref{lem:esd-simple-properties}, $G[F]$ contains a cycle. Since $G[F]$ is $2$-connected and contains a cycle, by Lemma \ref{lem:esd-three-nz}, $|NZ(F)|\ge 3$. 

\paragraph{Property (3).}  Let $A,B\in\mathcal L$ with
$A\subsetneq B$, and suppose for contradiction that
$x_v=0$ for every $v\in V(B)\setminus V(A)$.  Set
$H:=B\setminus A$ and $W':=V(B)\setminus V(A)$.  Since $A$ and $B$ are tight,
\[
        g_x(B)-g_x(A)=b(B)-b(A)=|H|-|W'|.
\]
On the other hand, because the vertices of $W'$ have zero $x$-value,
\[
        g_x(B)-g_x(A)
        =\sum_{uv\in H}(x_u+x_v)-x(W')
        =\sum_{uv\in H}(x_u+x_v).
\]

Let $p:=|\{(e,a):e\in H, a\in e\cap V(A)\}|$. We observe that $p=\sum_{a\in V(A)}d_H(a)=2|H[V(A)]| + |\delta_H(V(A), W')|$. We observe that $p>0$: Since $A\subsetneq B$, we have that $H\neq \emptyset$. If $W'=\emptyset$, then every edge of $H$ has both end-vertices in $V(A)$, so $p=2|H|>0$. If $W'\neq \emptyset$, then $G[B]$ is connected and since $A\neq \emptyset$, some edge of $B$ must have one end-vertex in $V(A)$ and another end-vertex in $W'$. Such an edge cannot belong to $A$ and hence, it lies in $H$ and hence, $\delta_H(V(A), W')\neq \emptyset$ showing that $p>0$. 

We have that 
\begin{align}
|H|-|W'|
&=\sum_{uv\in H}(x_u+x_v) \notag\\
&= \sum_{(e, a): e\in H, a\in e\cap V(H)}x_a \notag\\
&= \sum_{(e, a): e\in H, a\in e\cap V(A)}x_a \quad \quad \text{(since $x_u=0$ for all $u\in W'$)} \notag\\
&< \frac{p}{2}. \quad \quad \text{(since $x_u<1/2$ for all $u\in V$ and $p>0$)} \label{eq:esd-1}
\end{align}

Let $q:=|\{(e,w): e\in H, w\in e\cap W'\}|$. We observe that $q = \sum_{w\in W'}d_H(w)$. Since $V(A)$ and $W'$ partition $V(B)$, we have that $2|H|=|\{(e,h): e\in H, h\in V(H)\}|=|\{(e,a):e\in H, a\in e\cap V(A)\}| + |\{(e,w): e\in H, w\in e\cap W'\}|=p+q$. Now, consider a vertex $w\in W'$. Every edge of $B$ incident to $w$ must lie in $H=B\setminus A$ because no edge of $A$ can be incident to a vertex outside $V(A)$. Hence, $d_H(w) = d_B(w)$ for every $w\in W'$. Since $G[B]$ is $2$-connected, every vertex of $G[B]$ has degree at least $2$. Therefore, 
\[
q=\sum_{w\in W'}d_H(w) = \sum_{w\in W'}d_B(w) \ge 2|W'|. 
\]
Using $p+q=2|H|$, we get 
\[
2|H|=p+q \ge p + 2|W'|,
\]
a contradiction to $|H|-|W'|<p/2$. Thus, property (iii) holds. 
\end{proof}

\subsection{Counting Lemmas for Contradiction}
\paragraph{Notation.} Let $\mathcal{L}$ be the laminar family from Theorem \ref{thm:esd-basis}. 
From the basis structure, we have that $|\mathcal{L}|+|Z|=n$. We let $z$ denote $|Z|$ and use $p=n-z$ to denote the number of vertices $u$ with strictly positive $x_u$ value. 
Let $t$ denote the number of maximal sets in $\mathcal{L}$ with $\widehat{F}_1, \widehat{F}_2, \ldots, \widehat{F}_t$ being the maximal sets in $\mathcal{L}$. Since $G$ contains a cycle $C$, Lemma \ref{lem:esd-three-nz} implies that $p\ge 3$, and hence  $t\ge 1$. 
For each $i\in [t]$, let $\cB_i:=\{A\in \cL: A\subseteq \widehat{F}_i\}$ and $k_i:=|\cB_i|$. We have $p=|\mathcal{L}|=\sum_{i=1}^t k_i$. 

First consider the simple case when each maximal set has no other sets contained within it and moreover, the subgraph induced by the maximal-sets are pairwise vertex-disjoint. Thus, $\cB_i = \{\widehat{F}_i\}$ and $V(\widehat{F}_i)\cap V(\widehat{F}_j)=\emptyset$ for distinct $i, j\in [t]$. Then, a
contradiction is quite easy as follows: We have $p = |\cL| = t$ since
$k_i=1$ for each $i\in [t]$.  However, each $|NZ(\widehat{F}_i)| \ge 3$ by
Theorem~\ref{thm:esd-basis} which implies that
$\sum_{i=1}^t |NZ(\widehat{F}_i)| \ge 3t$. Since the subgraph induced by the maximal sets are
pairwise vertex-disjoint, we have that $p \ge \sum_{i=1}^t |NZ(\widehat{F}_i)|$ which is a
contradiction since $t \ge 1$. 

However, the assumption that the subgraph induced by the maximal-sets 
$\widehat{F}_1,\ldots,\widehat{F}_t$ are pairwise vertex-disjoint is too strong. In particular, a vertex
$v$ may belong to multiple maximal sets and we can no longer obtain a
contradiction. We need to account for the overlap of the sets even in this restricted case.
This motivates the following definition.

\begin{definition}[Global sharing loss]
For each vertex $v\in V-Z$, let
\[
        m_v:=|\{i\in[t]: v\in V(\widehat F_i)\}|,
\]
denote the number of maximal sets containing $v$. The global sharing loss is 
\[
        \sigma_{\mathrm{global}}:=\sum_{v\in V-Z}(m_v-1).
\]
\end{definition}

We observe that support vertices $v$ with $m_v=1$ contribute zero in the definition of $\sigma_{\mathrm{global}}$, so only support vertices shared between multiple blocks contribute positive amount to $\sigma_{\mathrm{global}}$. The following proposition is easy from the definitions. 

\begin{proposition}\label{prop:esd-intermed}
We have that 
\begin{align}
        \sum_{i=1}^t |NZ(\widehat F_i)|&=p+\sigma_{\mathrm{global}}.
\end{align}
\end{proposition}
\begin{proof}
We observe that 
\begin{align*}
\sum_{i=1}^\numberofblocks |NZ(\widehat F_i)| 
&= \sum_{i=1}^\numberofblocks \sum_{u\in V-Z} 1_{u\in V(\widehat{F}_i)} \\
&= \sum_{u\in V-Z}\sum_{i=1}^\numberofblocks 1_{u\in V(\widehat{F}_i)} = \sum_{u \in V-Z}m_v = \sum_{u\in V-Z}(m_v -1) + |V-Z| \\
& = \sigma_{\mathrm{global}} + \nnz. 
\end{align*}
\end{proof}

\begin{remark}
    For distinct $i, j \in [t]$, $\widehat{F}_i$ and $\widehat{F}_j$ are edge-disjoint, but $V(\widehat{F}_i)$ and $V(\widehat{F}_j)$ may share vertices. Moreover, every support vertex lies in at least one root $\widehat F_i$; otherwise all basis vectors would have zero coordinate at that vertex.  
\end{remark}

\paragraph{Two key lemmas.} The rest of the analysis is based on two key technical lemmas. We state the lemmas and use them to complete the proof of Theorem \ref{thm:esd-extreme-point} here. We will prove these two lemmas subsequently. 
The first bounds the sharing loss for a collection of tight sets.

\begin{lemma}[Union Feasibility Inequality]\label{lem:esd-feasibility}
Let $\mathcal A=\{A_1,\ldots,A_r\}\subseteq\mathcal T$ be a collection of
pairwise edge-disjoint tight sets.  Let $U:=\bigcup_{j=1}^r A_j$, and for
each $v\in V(U)$ let
\[
        m_v^{\mathcal A}:=|\{j\in[r]: v\in V(A_j)\}|.
\]
Set
\[
        \sigma(\mathcal A):=\sum_{v\in V(U)}(m_v^{\mathcal A}-1)
        =\sum_{j=1}^r |V(A_j)|-|V(U)|.
\]
If 
$r\ge 2$, then \[\sigma(\mathcal A)\le 2r-3.\]
\end{lemma}
We will see later in the proof of Theorem \ref{thm:esd-extreme-point} that we are interested in the sharing loss of support vertices only, but the above lemma is phrased in terms of the sharing loss of all vertices for notational ease. 
The second lemma builds on the preceding to show a surplus property. 
\begin{lemma}[Surplus in maximal sets]\label{lem:esd-surplus}
For each $i\in [t]$, we have that 
\[
        |\NZ(\widehat F_i)|\ge k_i+2.
\]
\end{lemma}

\begin{corollary}[Lower Bound on Global Sharing Loss]\label{coro:esd-lower}
$\sigma_{\mathrm{global}} \ge 2\numberofblocks$.
\end{corollary}
\begin{proof}
Combining Proposition \ref{prop:esd-intermed} 
and Lemma~\ref{lem:esd-surplus} gives
\[
\nnz + \sigma_{\mathrm{global}} = \sum_{i=1}^\numberofblocks |NZ(\widehat F_i)| \ge \sum_{i=1}^\numberofblocks (k_i + 2) = \nnz + 2\numberofblocks. 
\]
\end{proof}

\subsubsection{Proof of Theorem \ref{thm:esd-extreme-point}}
We complete the proof of Theorem \ref{thm:esd-extreme-point}. 
\begin{proof}[Proof of Theorem \ref{thm:esd-extreme-point}]

Assume for contradiction that $x_u<1/2$ for all $u\in V$. 

Since $G$ contains a cycle $C$, Lemma \ref{lem:esd-three-nz} implies that $p\ge 3$, and hence  $t\ge 1$. 
By Corollary \ref{coro:esd-lower}, we have that $\sigma_{\mathrm{global}}\ge 2t$.  

Suppose $t=1$. Then, there is only one maximal set $\widehat{F}_1$. Every support vertex lies in $V(\widehat{F}_1)$, hence $m_v=1$ for every support vertex $v$ and $\sigma_{\mathrm{global}}=0$, a contradiction.
Thus, we may assume that $t\ge 2$. The maximal sets $\widehat F_1,\ldots,\widehat F_t$ in $\mathcal{L}$ are pairwise edge-disjoint. By Lemma \ref{lem:esd-feasibility} applied to the collection $\mathcal{A}=\{\widehat F_1,\ldots,\widehat F_t\}$, we have that the total sharing loss $\sigma(\mathcal{A})\le 2t-3$.  Since $\sigma_{\mathrm{global}}$ counts the sharing loss of only support
vertices, while $\sigma(\mathcal{A})$ counts the sharing loss of all vertices in $\cup_{i\in [t]}V(\widehat{F}_i)$, we have that $\sigma_{\mathrm{global}}\le \sigma(\mathcal{A})\le 2t-3$, a contradiction. 
\end{proof}

\subsubsection{Proof of Lemma \ref{lem:esd-feasibility}}

\begin{proof}[Proof of Lemma \ref{lem:esd-feasibility}]
    Since the edge sets $A_j$ are pairwise edge-disjoint, the edges incident to a
vertex $v$ in the edge-union $U$ are partitioned among the sets $A_j$ containing
$v$.  With the convention $d_{A_j}(v)=0$ when $v\notin V(A_j)$, this gives
\[
        d_U(v)=\sum_{j=1}^r d_{A_j}(v)
        \qquad \forall v\in V(U).
\]
Starting from the definition of $g_x(U)$, we obtain
\[
\begin{aligned}
        g_x(U)
        &=\sum_{v\in V(U)}(d_U(v)-1)x_v \\
        &=\sum_{v\in V(U)}\left(\sum_{j=1}^r d_{A_j}(v)-1\right)x_v \\
        &=\sum_{j=1}^r\sum_{v\in V(A_j)} d_{A_j}(v)x_v
          -\sum_{v\in V(U)}x_v \\
        &=\sum_{j=1}^r\left(
            \sum_{v\in V(A_j)}(d_{A_j}(v)-1)x_v
            +\sum_{v\in V(A_j)}x_v\right)
          -\sum_{v\in V(U)}x_v \\
        &=\sum_{j=1}^r g_x(A_j)
          +\left(\sum_{j=1}^r\sum_{v\in V(A_j)}x_v
            -\sum_{v\in V(U)}x_v\right) \\
        &=\sum_{j=1}^r g_x(A_j)
          +\sum_{v\in V(U)}(m_v^{\mathcal A}-1)x_v.
\end{aligned}
\]

Similarly, from the definition of $b(U)$, pairwise edge-disjointness gives
$|U|=\sum_{j=1}^r |A_j|$, while the definition of $\sigma(\mathcal A)$ gives
\[
        \sigma(\mathcal A)=\sum_{v\in V(U)}(m_v^{\mathcal A}-1)
        =\sum_{j=1}^r |V(A_j)|-|V(U)|.
\]
Thus, $|V(U)|=\sum_{j=1}^r |V(A_j)|-\sigma(\mathcal A)$, and
\[
\begin{aligned}
        b(U)
        &=|U|-|V(U)|+1 \\
        &=\sum_{j=1}^r |A_j|-
          \left(\sum_{j=1}^r |V(A_j)|-\sigma(\mathcal A)\right)+1 \\
        &=\sum_{j=1}^r\bigl(|A_j|-|V(A_j)|+1\bigr)
          +\sigma(\mathcal A)-r+1 \\
        &=\sum_{j=1}^r b(A_j)+\sigma(\mathcal A)-r+1.
\end{aligned}
\]

Using tightness of every $A_j$ and feasibility of $U$, we get
\[
        0\le g_x(U)-b(U)
        =\sum_{v\in V(U)}(m_v^{\mathcal A}-1)x_v-(\sigma(\mathcal A)-r+1),
\]
which proves the first assertion.

Suppose $r\ge 2$. If $\sigma(\mathcal A)=0$, the bound is immediate.  Otherwise,
using $x_v<1/2$ for all $v$ gives
\[
        \frac{\sigma(\mathcal A)}2
        > \sum_{v\in V(U)}(m_v^{\mathcal A}-1)x_v
        \ge \sigma(\mathcal A)-r+1.
\]
Hence $\sigma(\mathcal A)<2r-2$, and since $\sigma(\mathcal A)$ is integral,
$\sigma(\mathcal A)\le 2r-3$.
\end{proof}

\subsubsection{Proof of Lemma \ref{lem:esd-surplus}}
For each $i\in[\numberofblocks]$, we recall that $k_i$ denotes the number of sets in  $\mathcal{B}_i:=\{A\in \cL: A\subseteq \widehat{F}_i\}$ and $V_i=NZ(\widehat{F}_i)$. 

\begin{proof}[Proof of Lemma \ref{lem:esd-surplus}]
    For a
set $S\in\mathcal B_i$, let $\mathcal T_S:=\{A\in \mathcal{B}_i: A\subseteq S\}$ and let $k_S:=|\mathcal T_S|$.  We recall that $\mathcal{B}_i$ is a laminar family. We prove that
\begin{align}
        |\NZ(S)|&\ge k_S+2.        \label{eq:esd-surplus}
\end{align}
by induction on the tree-representation of the laminar family corresponding to $\mathcal{B}_i$. 
If $S$ is a leaf, then $k_S=1$, and \eqref{eq:esd-surplus} follows from Lemma \ref{lem:esd-three-nz}.

Now suppose that $S$ has children $C_1,\ldots,C_r$.  The children are pairwise
edge-disjoint.  By induction, $|\NZ(C_j)|\ge k_{C_j}+2$ for each $j$, and
$k_S=1+\sum_{j=1}^r k_{C_j}$.

If $r=1$, property (iii) of Theorem \ref{thm:esd-basis} gives a support vertex in
$V(S)\setminus V(C_1)$.  Therefore
\[
        |\NZ(S)|\ge |\NZ(C_1)|+1
        \ge k_{C_1}+3=k_S+2.
\]

Assume next that $r\ge 2$.  Let
\[
        \sigma_{\NZ}:=\sum_{j=1}^r |\NZ(C_j)|
        -\left|\bigcup_{j=1}^r \NZ(C_j)\right|
\]
be the sharing loss among support vertices of the children.  This is at most
the all-vertex sharing loss 
\[
\sigma(\{C_1,\ldots,C_r\}):=\sum_{j=1}^r |V(C_j)|
        -\left|\bigcup_{j=1}^r V(C_j)\right|
\]
By Lemma \ref{lem:esd-feasibility} applied to $\mathcal{A}=\{C_1, \ldots, C_r\}$, we have that $\sigma_{\NZ}\le \sigma(\{C_1,\ldots,C_r\})\le 2r-3$.  Hence, 
\[
\begin{aligned}
        |\NZ(S)|
        &\ge \left|\bigcup_{j=1}^r \NZ(C_j)\right|  
        = \sum_{j=1}^r |\NZ(C_j)|-\sigma_{\NZ} \\
        &\ge \sum_{j=1}^r (k_{C_j}+2)-(2r-3) 
        = \sum_{j=1}^r k_{C_j}+3
         = k_S+2.
\end{aligned}
\]
This proves \eqref{eq:esd-surplus}, and applying it to the root $S=\widehat F_i$ proves the lemma.
\end{proof}

\subsection{Iterative Rounding via Edge Strong Density Polyhedron}\label{sec:esd-rounding-algorithm}
Theorem \ref{thm:esd-extreme-point} and Lemma \ref{lem:esd-supermod} together imply  
a $2$-approximation for FVS via standard application of iterative rounding. 
For an input graph $G=(V, E)$ with vertex-costs $c:V\rightarrow \mathbb{R}_{\ge 0}$, we recall that $\min\{c^Tx: x\in P_{\ESD}(G)\cap \mathbb{Z}^V\}$ is a valid formulation of FVS. 

We now state the iterative rounding algorithm. For input graph $G=(V, E)$ with vertex-costs $c:V\rightarrow \mathbb{R}_{\ge 0}$, repeat the following while $G$ has at least one cycle: (1) Compute an extreme point optimum solution $x$ for $\min\{c^Tx: x\in P_{\ESD}(G)\}$---this can be done in polynomial time via Ellipsoid since the separation oracle can be implemented via submodular minimization because of Lemma \ref{lem:esd-supermod}. 
(2) By Theorem \ref{thm:esd-extreme-point}, there exists a vertex $u\in V$ such that $x_u\ge 1/2$; include the vertex $u$ in the solution and remove it from the graph $G$. The approximation factor of the solution constructed by this procedure relative to the starting extreme point optimum
solution of the LP is at most $2$ via standard iterative rounding analysis. 

The above-mentioned iterative rounding algorithm relies on solving an LP with exponential number of constraints. In Section \ref{sec:ESD-and-orient-equivalence}, we design an iterative rounding algorithm based on a different LP that relies on solving an LP with polynomial number of constraints. 

\section{Orientation Polyhedron }\label{sec:ESD-and-orient-equivalence}
Chandrasekaran, Chekuri, Fiorini, Kulkarni, and Weltge \cite{CCFKW25} gave an orientation-based extended formulation for FVS. In this section, we relate their formulation to the edge strong density polyhedron. We next exploit this connection and the extreme point result for edge strong density polyhedron to design an alternative $2$-approximation via iterative rounding that involves solving an LP with polynomial number of constraints. 

The orientation-based formulations for FVS and PFDS in \cite{CCFKW25} are based on the connection between these problems and the well-known Densest Subgraph Problem (DSG) that we mentioned in the related work section. FVS and PFDS can be viewed as density deletion problems, and orientation-based formulations arise by considering LP relaxations for DSG that were described by Charikar \cite{charikar_greedy_2000}. We do not spell out the intuition and details of these connections here and instead refer interested readers to \cite{CCFKW25} and to a more recent work on the approximability of density deletion more broadly \cite{ChandrasekaranCK2025}. 

\subsection{Edge Strong Density Polyhedron and Orientation Polyhedron }
In this section, we show that the edge strong density polyhedron is equivalent to an orientation based extended formulation for FVS that was given in \cite{CCFKW25}. This leads to an alternative $2$-approximation via iterative rounding for FVS. 
\begin{definition}[Orientation Polyhedron]
Let $G=(V\cup I,E)$ be a graph containing at least one cycle, where $V$ is the set of non-isolated vertices and $I$ is the set of isolated vertices of $G$. 
We define $Q_{\orient}'(G)$ to be the set of points $(x, y)$ 
satisfying the following system: 
\begin{align}
x_v + x_w + y^f_{e,v} + y^f_{e,w} &\ge 1 \quad \forall e = vw \in E, f \in E \label{eq:orient_1} \\
x_v + \sum_{e=vw \in E} y^f_{e,w} &\ge 1 \quad \forall v\in V, \text{ edge } f \in E \label{eq:orient_2} \\
\sum_{v \in V \setminus \{a,b\}} x_v + \sum_{e=vw \in E \setminus \{f\}} (y^f_{e,w} + y^f_{e,v}) &\le |V| - 2 \quad \forall f = ab \in E \label{eq:orient_3}\\
y^f_{e,v} &\ge 0 \quad \forall e \in \delta(v), v \in V, f \in E \label{eq:orient_4}\\
0\le x_u &\le 1 \quad \forall u \in V. \label{eq:orient_5}
\end{align}
The \emph{orientation polyhedron} $Q_{\text{orient}}(G)$ is defined as the projection of $Q'_{\text{orient}}(G)$ to the $x$ variables (equivalently, as the set of points $x\in [0,1]^V$ such that there exist variables $y^f_{e,v}$ for every $e = vw \in E$ and $f \in E$ satisfying \eqref{eq:orient_1}---\eqref{eq:orient_4}). 
\end{definition}

\begin{theorem} \label{thm:equiv}
Let $G=(V,E)$ be a graph containing at least one cycle and no isolated vertices. Then, 
\[P_{\ESD}(G) = Q_{\mathrm{orient}}(G). \]

\end{theorem}

For ease of notation, we denote $P_{\ESD}:=P_{\ESD}(G)$ and $Q_{\text{orient}}:=Q_{\text{orient}}(G)$. In Lemma \ref{lem:orient-is-in-edge-SD} below, we show that $Q_{\mathrm{orient}}\subseteq P_{\mathrm{Edge-SD}}$ by showing that the inequalities of $P_{\mathrm{Edge-SD}}$ are implies by that of $Q_{\mathrm{orient}}$. 
In Lemma \ref{lem:edge-SD-in-orient} below, we show that $P_{\mathrm{Edge-SD}} \subseteq Q_{\mathrm{orient}}$ via Farkas Lemma: in particular, we show that a certificate that violates the orientation system can be rounded to an integral certificate, which corresponds exactly to an edge-subset inequality of $P_{\mathrm{Edge-SD}}$. 
Lemmas \ref{lem:orient-is-in-edge-SD} and \ref{lem:edge-SD-in-orient} together prove Theorem \ref{thm:equiv}. 

\begin{lemma}\label{lem:orient-is-in-edge-SD}
    $Q_{\mathrm{orient}}\subseteq P_{\mathrm{Edge-SD}}$. 
\end{lemma}
\begin{proof}
Let $x \in Q_{\text{orient}}$, so there exist $y \ge 0$ satisfying the orientation constraints. Let $F \subseteq E$ be an arbitrary non-empty subset of edges. We need to show that $x$ satisfies the Edge-SD constraint for $F$. Fix an arbitrary edge $f = ab \in F$.

We have that 
\begin{equation}
|F| - |V(F)| - \sum_{u \in V(F)} (d_F(u) - 1) x_u = \sum_{e=vw \in F} (1 - x_v - x_w) + \sum_{u \in V(F)} (x_u - 1). \label{eq:part1_lhs}
\end{equation}

By constraint \eqref{eq:orient_1}, for every edge $e = vw \in F \setminus \{f\}$, we have $1 - x_v - x_w \le y^f_{e,v} + y^f_{e,w}$. 
Thus,
\begin{equation}
\sum_{e=vw \in F} (1 - x_v - x_w) \le (1 - x_a - x_b) + \sum_{e=vw \in F \setminus \{f\}} (y^f_{e,v} + y^f_{e,w}). \label{eq:bound1}
\end{equation}

Next, by constraint \eqref{eq:orient_2}, for every vertex $u \in V$, we have $x_u - 1 + \sum_{e=uw \in E} y^f_{e,w} \ge 0$. 
Hence, 
\begin{align}
\sum_{u \in V(F)} (x_u - 1) &\le \sum_{u \in V(F)} (x_u - 1) + \sum_{u \in V-V(F)} \left(x_u - 1 + \sum_{e=uw \in E} y^f_{e,w}\right) \nonumber \\
&= \sum_{u \in V} (x_u - 1) + \sum_{u \in V-V(F)} \sum_{e=uw \in E} y^f_{e,w}. \label{eq:bound2}
\end{align}

Adding the bounds \eqref{eq:bound1} and \eqref{eq:bound2} yields an upper bound for \eqref{eq:part1_lhs}:
\begin{align*}
&|F| - |V(F)| - \sum_{u \in V(F)} (d_F(u) - 1) x_u \\
&\quad \quad = \sum_{e=vw \in F} (1 - x_v - x_w) + \sum_{u \in V(F)} (x_u - 1)\\
&\quad \quad\le (1 - x_a - x_b) + \sum_{u \in V} (x_u - 1) + \sum_{e=vw \in F \setminus \{f\}} (y^f_{e,v} + y^f_{e,w}) + \sum_{u \in V- V(F)} \sum_{e=uw \in E} y^f_{e,w} \\
&\quad \quad\le (1 - x_a - x_b) + \sum_{u \in V} (x_u - 1) + \sum_{e=vw \in E \setminus \{f\}} (y^f_{e,v} + y^f_{e,w}) \\
&\quad \quad= 1 - |V| + \sum_{u \in V \setminus \{a,b\}} x_u + \sum_{e=vw \in E \setminus \{f\}} (y^f_{e,v} + y^f_{e,w})\\
&\quad \quad\le 1-|V| + |V|-2 \quad \quad \text{(by constraint \eqref{eq:orient_3})}\\
&\quad \quad = -1
\end{align*}
where the second inequality follows because $y \ge 0$. 
Therefore, we have, 
\[ |F| - |V(F)| - \sum_{u \in V(F)} (d_F(u) - 1) x_u \le -1. \]
\end{proof}

\begin{lemma}\label{lem:edge-SD-in-orient}
    $P_{\mathrm{Edge-SD}} \subseteq Q_{\mathrm{orient}}$. 
\end{lemma}
\begin{proof}
Let $x \in P_{\text{Edge-SD}}$. We will use Farkas' Lemma to show that there exists $y \ge 0$ satisfying constraints \eqref{eq:orient_1}, \eqref{eq:orient_2}, and \eqref{eq:orient_3}. The orientation constraints for $y^f$ are on disjoint sets of variables, hence it suffices to show that for each fixed edge $f=ab\in E$, there exist $y^f_{e,v}$ for every $e=vw\in E$ and $f\in E$ satisfying the orientation constraints \eqref{eq:orient_1}, \eqref{eq:orient_2}, and \eqref{eq:orient_3}. Fix an arbitrary edge $f = ab \in E$. 

By Farkas' Lemma, there exist $y_{e,v} \ge 0$ for every $e=vw\in E$ satisfying \eqref{eq:orient_1}, \eqref{eq:orient_2}, and \eqref{eq:orient_3} if and only if for all $(p, q, z) \ge 0$ satisfying 
\begin{align}
p_e + q_w &\le z \quad \forall e\in E \setminus \{f\}, w\in e \label{eq:dual1}\\
p_f + q_a &\le 0 \label{eq:dual2} \\
p_f + q_b &\le 0 \label{eq:dual3}
\end{align}
we have that $\text{Obj}(p,q,z)\le 0$, where 
\[ \text{Obj}(p,q,z) = \sum_{e=vw \in E} p_e (1 - x_v - x_w) + \sum_{v \in V} q_v (1 - x_v) - z \left( |V| - 2 - \sum_{v \in V \setminus \{a,b\}} x_v \right). \]

Let $(p, q, z)\ge 0$ satisfy \eqref{eq:dual1}--\eqref{eq:dual3} such that $\text{Obj}(p, q, z)$ is maximized.  
We will show that $\text{Obj}(p,q,z)\le 0$. For the sake of contradiction, suppose $\text{Obj}(p,q,z)> 0$. If $z=0$, then $p_e=q_w = 0$ for every $e\in E$ and $w\in V$ and hence, $\text{Obj}(p, q, z)=0$. Hence, we may assume that $z>0$. 

Consider the following LP: 
\begin{equation}
\alpha:=\max\{\text{Obj}(p', q', 1): (p', q', 1)\text{ satisfy }\eqref{eq:dual1}-\eqref{eq:dual3}\} \label{eq:simple-dual}
\end{equation}
We have that $(p/z, q/z, 1)$ is an optimum solution to the LP \eqref{eq:simple-dual} and moreover, $\text{Obj}(p,q,z)=\alpha z$. Thus, it suffices to show that $\alpha\le 0$. 

First, we show that the LP \eqref{eq:simple-dual} 
has an integral optimum solution. For this, we observe that the constraints of the LP simplify to the following:
\begin{align}
p_e + q_w &\le 1 \quad \forall e \in E \setminus \{f\}, w\in e \label{eq:simple-dual1}\\
p_f + q_a &\le 0 \label{eq:simple-dual2} \\
p_f + q_b &\le 0 \label{eq:simple-dual3}\\
p_e&\ge 0\ \forall\ e\in E\\
q_u&\ge 0\ \forall\ u\in V
\end{align}
This constraint matrix is totally unimodular: the constraint matrix is the node-edge incidence matrix
of the bipartite graph with one side $E$ for the $p_e$ variables and the other side $V$ for the $q_v$ variables, with an incidence row $(e, v)$ whenever $v\in e$.
Moreover, the RHS is integral. Hence, the LP \eqref{eq:simple-dual} has an integral optimum solution $(\bar{p}, \bar{q})$, where $\bar{p}\in \{0,1\}^E$ and $\bar{q}\in \{0, 1\}^V$. 

We now show that $\alpha \le 0$. 
Let $F':=\{f\}\cup \{e\in E:\bar{p}_e = 1\}$. Then, by constraints \eqref{eq:simple-dual1}--\eqref{eq:simple-dual3}, we have that $\bar{q}_v = 0$ for all $v \in V(F')$. 
Since $x_v\le 1$ for every $v\in V$ and $(\bar{p}, \bar{q}, 1)$ maximizes $\text{Obj}(\bar{p},\bar{q},1)$ while satisfying \eqref{eq:simple-dual1}, we may choose an integral optimum with $\bar{q}_v = 1$ for all $v \in V-V(F')$. 
Substituting these integral values into the objective, we have that 
\begin{align*}
&\text{Obj} (\bar{p}, \bar{q}, 1) \\
&\quad \quad = \sum_{e \in F' \setminus \{f\}} (1 - x_v - x_w) + \sum_{v\in V- V(F')} (1 - x_v) - \left( |V| - 2 - \sum_{v \in V \setminus \{a,b\}} x_v \right) \\
&\quad \quad = (|F'| - 1) - \sum_{e \in F' \setminus \{f\}} (x_v + x_w) + |V| - |V(F')| - \sum_{v \in V- V(F')} x_v - |V| + 2 + \sum_{v \in V} x_v - x_a - x_b \\
&\quad \quad = |F'| - |V(F')| + 1 - \sum_{e \in F' \setminus \{f\}} (x_v + x_w) + \sum_{v \in V(F')} x_v - x_a - x_b\\
&\quad \quad = |F'| - |V(F')| + 1 - \sum_{e \in F' } (x_v + x_w) + \sum_{v \in V(F')} x_v\\
&\quad \quad = |F'| - |V(F')| + 1 - \sum_{v\in V(F')} d_{F'}(v) x_v + \sum_{v \in V(F')} x_v\\
&\quad \quad = |F'| - |V(F')| + 1 - \sum_{v\in V(F')} (d_{F'}(v)-1) x_v \\
&\quad \quad \le 0, 
\end{align*}
where the last inequality is because of the Edge-SD constraint for $F'$ since $F'$ is non-empty. 

\end{proof}

\subsection{Iterative Rounding via Orientation Polyhedron}
\label{sec:iter-orient}
In this section, we design a $2$-approximation via iterative rounding with respect to the orientation polyhedron. We note that although $P_{\ESD}(G) = Q_{\text{orient}}(G)$, extreme point property for $P_{\ESD}(G)$ does not directly hold for $Q'_{\text{orient}}(G)$. 
This is because $Q'_{\text{orient}}(G)$ is an extended formulation of $P_{\ESD}(G)$ --- the extended space could have extreme points where none of the original variables are large. However, Theorems \ref{thm:esd-extreme-point} and \ref{thm:equiv} do imply that there is some extreme point in the extended space where some original variable is large. We summarize this in the following corollary. We encourage the readers to compare and contrast Theorem \ref{thm:esd-extreme-point} with the following corollary. 

\begin{corollary}\label{coro:orient-extreme-point}
Let $G$ be a graph containing a cycle with vertex costs $c: V\rightarrow \mathbb{R}_{\ge 0}$ and no isolated vertices. Then, there exists an extreme point optimum solution $x$ to $\min\{c^Tx: (x,y)\in Q'_{\text{orient}}(G)\}$ such that at least one of the coordinates of $x$ has value at least $1/2$. 
\end{corollary}

Next, we use this corollary to design a $2$-approximation via iterative rounding with respect to the orientation polyhedron. For an input graph $G=(V, E)$ with vertex-costs $c:V\rightarrow \mathbb{R}_{\ge 0}$, it is known that $\min\{c^Tx: (x,y)\in Q'_{\text{orient}}(G), x \text{ and }y \text{ are integral}\}$ is a valid formulation of FVS \cite{CCFKW25}. We note that $\min\{c^Tx: (x,y)\in Q'_{\text{orient}}(G)\}$ is a polynomial-sized LP and hence, can be solved in polynomial time. 

We now state the iterative rounding algorithm. For input graph $G=(V, E)$ with no isolated vertices and with vertex-costs $c:V\rightarrow \mathbb{R}_{\ge 0}$, repeat the following while $G$ has at least one cycle: (1) 
For each $u\in V$, compute an extreme point optimum $x$ for $\min\{c^Tx: (x,y)\in Q'_{\text{orient}}(G), x_u\ge 1/2\}$; among all computed extreme point optima, pick the one with the cheapest objective value; let $u$ be the vertex associated with such an extreme point optima. 
(2) Include the vertex $u$ in the solution and remove it from the graph $G$; finally, remove all isolated vertices. 

The enumeration over constraints $x_u \ge 1/2$ avoids relying on an extreme point of the extended space. Since
$P_{\ESD}(G)=Q_{\text{orient}}(G)$ and $P_{\ESD}(G)$ has an optimal projected extreme point with
some coordinate at least $1/2$, at least one enumerated LP has the same optimum value
as the unconstrained LP. The approximation factor of the solution constructed by this procedure relative to the starting extreme point optimum solution of the LP is at most $2$ via standard iterative rounding analysis based on Corollary \ref{coro:orient-extreme-point}. 
\section{Conclusion}\label{sec:conclusion}
We conclude with a few interesting directions for future work.
As we remarked earlier, LP-based approximations via solvable LPs are valuable.
From this perspective, a computational study to explore instance-based approximations relative to our LP, namely $P_{\ESD}(G)$ or $Q_{\text{orient}}(G)$, would be informative, and we leave it for future work. A bottleneck in the computational study is LP-solving,  which can be rather slow owing to the large number of constraints for both these polyhedra. This naturally leads to the question of designing fast algorithms to solve these LPs exactly or near-optimally. It is also of interest to find more efficient ways to round fractional solutions since iterative rounding is typically slow.

The extreme point conjecture regarding strong density polyhedron for FVS was itself formulated as a stepping stone towards better approximations for two related problems: Treewidth Deletion and Subset Feedback Vertex Set (SFVS). 
A recent work showed a randomized constant-approximation for Treewidth Deletion \cite{wlodarczyk2025losing}. It would be interesting to generalize the LP-based approaches for FVS to Treewidth Deletion to improve on the constant factor as well as to design a deterministic constant-approximation. 
The precise approximability of SFVS is still undetermined with the best-known upper bound being $8$ \cite{EvenNSZ00} and the best-known lower bound being $2$ coming from FVS. 
Chekuri and Madan \cite{chekuri-madan16} described a compact LP relaxation for SFVS with an integrality gap of at most $13$, and it remains open to improve this upper bound.

\paragraph{Acknowledgements.} For AI disclosure, see Section \ref{sec:techniques-and-ai-disclosure}. 

\bibliographystyle{abbrv}
\bibliography{references}
\appendix

\section{Iterative Rounding via Strong Density Polyhedron}\label{sec:rounding-algorithm}
The key challenge in directly using the extreme point result for $P_{SD}(G)$ (Theorem \ref{thm:extreme-point}) to obtain a $2$-approximation for FVS via iterative rounding is that optimizing over $P_{SD}(G)$ is not known to be polynomial time. In particular, the strong density polyhedron has exponentially many constraints and no efficient separation oracle is known. In this section, we leverage the extreme point result (Theorem \ref{thm:extreme-point}) and the conditional supermodularity of $f_x$ (Lemma~\ref{lem:supermod}) with the ellipsoid algorithm to achieve a polynomial-time $2$-approximation. We use notation and terminology that were defined in Section \ref{sec:extreme-point}.

\subsection{Separation below the $1/2$-threshold}

The key algorithmic fact is that once all remaining coordinates are strictly below $1/2$, violated strong-density inequalities can be found in polynomial time.

\begin{lemma}[Polynomial-time separation below $1/2$]\label{lem:sep-low}
Let $H=(W,F)$ be a graph, and let $x\in \R^W_{\ge 0}$ satisfy $x_u<1/2$ for every $u\in W$. Then, there exists a polynomial-time algorithm to
\begin{itemize}
  \item[(a)] either certify that $x\in P_{\SD}(H)$,
  \item[(b)] or return a set $S\subseteq W$ with $E[S]\neq \emptyset$ such that the strong density constraint corresponding to $S$ is violated by $x$.
\end{itemize}
\end{lemma}

\begin{proof}
By Proposition \ref{prop:sd-ineqs-equiv}, a violated strong-density inequality is exactly a set $S\subseteq W$ with $E_H[S]\neq\emptyset$ and $f_x(S)>0$. By Lemma~\ref{lem:supermod}, $f_x$ is supermodular, hence $-f_x$ is submodular.

For each edge $e=ab\in F$, define a set function on $h:2^{W\setminus\{a,b\}}\rightarrow \R$ by
\[
h_{x,e}(T):=-f_x(T\cup\{a,b\})\ \forall\ T\subseteq W\setminus \{a,b\}.
\]
Since $-f_x$ is submodular, $h_{x,e}$ is also submodular. Therefore, $h_{x,e}$ can be minimized in polynomial time by a submodular-function minimization algorithm. Let $T_e$ be a minimizer and set
$S_e:=T_e\cup\{a,b\}$. Then, $S_e$ maximizes $f_x(S)$ among all sets $S\subseteq V$ such that $e\in E[S]$. 

Let $S^*$ be a subset maximizing $f_x(S)$ over all $S\subseteq W$ with $E_H[S]\neq\emptyset$. Since $E_H[S^*]\neq\emptyset$, the set $S^*$ contains some edge $e$. By the choice of $S_e$,
$f_x(S_e)\ge f_x(S^*)$. 
Thus, $\max\{f_x(S):E_H[S]\neq\emptyset\} = \max_{e\in F} f_x(S_e)$.

Hence, after solving one submodular minimization problem per edge, we know whether the maximum of $f_x$ over all relevant subsets is positive. If the maximum is positive, the corresponding set $S_e$ yields a violated strong-density inequality by Lemma \ref{lem:supermod}(2). If the maximum is at most $0$, then Lemma \ref{lem:supermod}(2) implies $x\in P_{\SD}(H)$.
\end{proof}

\begin{corollary}[Residual Separation oracle]\label{cor:residual-sep}
There is a polynomial-time algorithm that takes a graph $H=(W,F)$ and $x\in \R^W_{\ge 0}$ as inputs where $A:=\{u\in W:x_u\ge 1/2\}$ and 
\begin{itemize}
  \item[(a)] either certifies that $x|_{W\setminus A}\in P_{\SD}(H-A)$
  \item[(b)] or returns a set $S\subseteq W\setminus A$ with $E_{H-A}[S]\neq \emptyset$ such that the strong density constraint corresponding to $S$ is violated by $x$.
\end{itemize}
\end{corollary}

\begin{proof}
Every coordinate of $x|_{W\setminus A}$ is $<1/2$, so the corollary follows by applying Lemma \ref{lem:sep-low} to the residual graph $H-A$ and vector $x|_{W\setminus A}$.
\end{proof}

\subsection{Produce-Witness Subroutine}
The following lemma shows that the ellipsoid method can be used to find a useful fractional solution even though we cannot efficiently separate over $P_{SD}(G)$ in general. We emphasize that the third property below will be crucial in the iterative rounding-based $2$-approximation later. 

\begin{lemma}[Produce-Witness]\label{lem:poly-time-witness}
Let $G = (V, E)$ be a graph containing a cycle, with non-negative vertex costs $c$. There is a polynomial-time algorithm that finds a vector $x^* \in \mathbb{R}^V_{\ge 0}$ such that:
\begin{enumerate}
    \item[\emph{(i)}] $c^T x^* \le \emph{LP}_{SD}(G):=\min\{c^Tx: x\in P_{\SD}(G)\}$,
    \item[\emph{(ii)}] there exists $u \in V$ with $x^*_u \ge 1/2$, and
    \item[\emph{(iii)}] letting $A := \{u \in V : x^*_u \ge 1/2\}$, the restriction of $x^*$ to $V - A$ lies in $P_{SD}(G - A)$.
\end{enumerate}
\end{lemma}

\begin{proof}
We define a weak separation oracle and use the ellipsoid method.

\medskip
\noindent\textbf{Weak separation oracle $\mathcal{O}$.}
Given $x^* \in \mathbb{R}^V_{\ge 0}$, let $A = \{u \in V : x^*_u \ge 1/2\}$ and let $x'$ be the restriction of $x^*$ to $V - A$. 
We then invoke the residual separation oracle from Corollary \ref{cor:residual-sep} on graph $G$ and vector $x^*$. 
\begin{itemize}
    \item If $x' \in P_{SD}(G - A)$: declare $x^*$ \emph{accepted}.
    \item If a violated constraint is found, i.e., a set $S \subseteq V - A$ with $f_{x'}(S) > 0$: return the halfspace $\{x : g_x(S) \ge b(S)\}$.
\end{itemize}

\medskip
\noindent\textbf{Key property of $\mathcal{O}$.}
Every halfspace returned by $\mathcal{O}$ is valid for $P_{SD}(G)$. In particular, if $x^* \in P_{SD}(G)$, then $\mathcal{O}$ accepts (since $x' \in P_{SD}(G-A)$). Therefore, $\mathcal{O}$ never rejects a point in $P_{SD}(G)$.

\medskip
\noindent\textbf{Running the ellipsoid method.}
Apply the ellipsoid method to solve $\min c^T x$ over $P_{SD}(G)$, using $\mathcal{O}$ as the separation oracle. The ellipsoid method queries $\mathcal{O}$ at the centers of successive ellipsoids. Since every returned halfspace contains $P_{SD}(G)$ and $P_{SD}(G) \neq \emptyset$, the standard volume argument guarantees that within polynomially many steps, either the ellipsoid method concludes that $P_{SD}(G) \cap \{c^T x \le \gamma\}$ is empty (for the current binary search threshold $\gamma$), or $\mathcal{O}$ accepts a point $\hat{x}$ with $c^T \hat{x} \le \gamma$.

By the equivalence of separation and optimization (Gr\"otschel--Lov\'asz--Schrijver), this procedure finds in polynomial time a point $\hat{x}$ that $\mathcal{O}$ accepts, with $c^T \hat{x} \le \text{LP}_{SD}(G)$. By construction, the restriction of $\hat{x}$ to $V - A_{\hat{x}}$ lies in $P_{SD}(G - A_{\hat{x}})$, so $\hat{x}$ satisfies properties~(i) and~(iii).

\medskip
\noindent\textbf{Ensuring property~(ii).}
If $A_{\hat{x}} \neq \emptyset$ (some $\hat{x}_u \ge 1/2$), then $\hat{x}$ satisfies all three properties and we are done.

If $A_{\hat{x}} = \emptyset$ (all $\hat{x}_u < 1/2$), then the acceptance of $\hat{x}$ by $\mathcal{O}$ means $\hat{x} \in P_{SD}(G)$. Since $c^T \hat{x} \le \text{LP}_{SD}(G)$ and $\hat{x} \in P_{SD}(G)$, the point $\hat{x}$ is optimal for $P_{SD}(G)$. We now find an \emph{extreme point} of $P_{SD}(G)$.

By Lemma \ref{lem:sep-low}, we have an \emph{exact} polynomial-time separation oracle for $P_{SD}(G)$ at every point in which all coordinates are less than $1/2$. We use this to find a vertex of $P_{SD}(G)$ as follows. Starting from $\hat{x}$, we solve a sequence of at most $|V|$ auxiliary LPs over $P_{SD}(G)$, each with a different objective chosen to reduce the dimension of the face on which the current point lies, until we reach a vertex. This is the standard procedure for finding a vertex of a polyhedron given a feasible point and a separation oracle; see Chapter 6 of Grotschel-Lovasz-Schrijver. 

To solve each auxiliary LP, we apply the ellipsoid method on $P_{SD}(G)$ using the weak oracle $\mathcal{O}$. If at any point during these ellipsoid runs $\mathcal{O}$ accepts a center point $x_t$ with $A_{x_t} \neq \emptyset$, then $x_t$ has a coordinate $\ge 1/2$, the restriction property holds (since $\mathcal{O}$ accepted), and $c^T x_t \le c^T \hat{x} = \text{LP}_{SD}(G)$ (since we optimize over $P_{SD}(G) \cap \{c^T x \le c^T \hat{x}\}$). So $x_t$ satisfies all three properties and we stop.

If no such early termination occurs, then every point accepted by $\mathcal{O}$ during the procedure has all coordinates less than $1/2$, and hence lies in $P_{SD}(G)$. In this case, $\mathcal{O}$ acts as an exact separation oracle for $P_{SD}(G)$ throughout the procedure, and the vertex-finding procedure correctly returns a vertex $\bar{x}$ of $P_{SD}(G)$. By Theorem~\ref{thm:extreme-point} (applicable since $G$ contains a cycle), $\bar{x}$ has a coordinate $\ge 1/2$. Since $\bar{x} \in P_{SD}(G)$, the restriction of $\bar{x}$ to $V - A_{\bar{x}}$ lies in $P_{SD}(G - A_{\bar{x}})$, and $c^T \bar{x} = \text{LP}_{SD}(G)$. So $\bar{x}$ satisfies all three properties.

In both cases, the algorithm runs in polynomial time and returns a vector satisfying properties (i), (ii), and (iii).
\end{proof}

\subsection{The Algorithm}
In this section, we exploit the produce-witness subroutine from the previous section to design an iterative rounding based $2$-approximation for FVS. Our algorithm is stated in Algorithm \ref{alg:fvs}. 
\begin{algorithm}[H]
\caption{Iterative Rounding for Feedback Vertex Set}
\label{alg:fvs}
\begin{algorithmic}[1]
\State \textbf{Input:} Graph $G = (V, E)$ with non-negative vertex costs $c: V\rightarrow \R_{\ge 0}$.
\State \textbf{Output:} A feedback vertex set $R$. 
\State Initialize $H_0\gets G$, $R_0 \gets \emptyset$, $i\gets 0$
\While{$H_i$ contains a cycle} 
    \State $x^{(i)}:=$ Produce-Witness$(H_i, c)$ 
    \State $A_i:=\{u\in V(H_i): x_u^{(i)}\ge 1/2\}$
    \State $R_{i+1}:=R_i\cup A_i$ and $H_{i+1} := H_i - A_i$
    \State $i\gets i+1$
\EndWhile
\State \Return $R_{i}$
\end{algorithmic}
\end{algorithm}

We now analyze the run-time of the algorithm. 

\begin{lemma}\label{lem:run-time}
Algorithm \ref{alg:fvs} can be implemented to run in polynomial time. 
\end{lemma}
\begin{proof}
We first argue that the algorithm terminates in at most $|V|$ iterations. By property (ii) of Lemma \ref{lem:poly-time-witness}, each call to Produce-Witness returns a non-empty set $A_i$ and hence, each iteration deletes at least one vertex from the current graph. Thus, the algorithm terminates in at most $|V|$ iterations. 

By Lemma \ref{lem:poly-time-witness}, each call to Produce-Witness can be implemented in polynomial time and hence, each iteration can be implemented in polynomial time. So,  the algorithm can be implemented to run in polynomial time. 
\end{proof}

Next, we analyze the approximation factor of the algorithm. 
\begin{lemma}\label{lem:apx-factor}
    Algorithm \ref{alg:fvs} returns a feedback vertex set $R$ with $\sum_{u \in R}c_u \le 2 \OPT$. 
\end{lemma}
\begin{proof}
    By Lemma \ref{lem:run-time}, the algorithm terminates. Suppose that the algorithm terminates with $i=t$. Let $R:=R_{t}$ denote the set returned by the algorithm. By termination criteria, the set $R$ is a feedback vertex set of $G$. We bound the approximation factor. 
    
    Let $z_i:=\min\{c^Tx: x\in P_{\SD}(H_i)\}$ for each $i\in \{0, 1, \ldots, t\}$. Since $H_t$ is a forest, we have that $z_t = 0$ (since the zero vector is feasible and optimal). 

    Fix an iteration $i\in \{0, 1, 2, \ldots, t-1\}$. Let $(\Fset_i, x^{(i)}, A_i)$ denote the tuple returned by Produce-Witness$(H_i, c)$. Let 
    \[
    \phi_i:=\sum_{u\in V(H_i)} c_u x_u^{(i)}. 
    \]
    We have that $\phi_i\le z_i$ by property (i) of Lemma \ref{lem:poly-time-witness}. 
    We know that $x^{(i)}|_{V(H_i)\setminus A_i}\in P_{\SD}(H_i-A_i)=P_{\SD}(H_{i+1})$ by property (iii) of Lemma \ref{lem:poly-time-witness}. Therefore, 
    \[
    z_{i+1}\le \sum_{u\in V(H_i)\setminus A_i}c_u x_u^{(i)} = \phi_i - \sum_{u\in A_i}c_u x_u^{(i)}. 
    \]
    Since every $u\in A_i$ has $x_u^{(i)}\ge 1/2$, we have that 
    \[
    \frac{1}{2}\sum_{u\in A_i}c_u \le \sum_{u\in A_i}c_u x_u^{(i)}. 
    \]
    Combining the last two inequalities, we obtain that 
    \[
    \sum_{u\in A_i} c_u \le 2(\phi_i - z_{i+1}) \le 2(z_i - z_{i+1}). 
    \]

    Summing over all iterations, we obtain that 
    \[
    \sum_{u\in R_i}c_u = \sum_{i=0}^{t-1}\sum_{u\in A_i}c_u \le 2\sum_{i=0}^{t-1} (z_i - z_{i+1}) = 2z_0. 
    \]
    We know that $z_0=\min\{c^Tx: x\in P_{\SD}(G)\}\le \OPT$. Hence, $\sum_{u\in R_i}c_u\le 2\OPT$.
\end{proof}

\end{document}